\documentclass[a4paper,11pt,reqno]{amsart}

\usepackage[utf8]{inputenc}
\usepackage{mathrsfs,dsfont,hyperref,amsmath,amssymb,amsthm,amsfonts,amstext,amsopn,amsxtra,mathrsfs,esint,enumitem,bookmark,mathtools}
\usepackage[capitalize]{cleveref}
\usepackage{braket}

\newtheorem{lemma}{Lemma}
\newtheorem{theorem}[lemma]{Theorem}

\newtheorem{corollary}[lemma]{Corollary}
\newtheorem{remark}[lemma]{Remark}

\newcommand{\R}{\mathbb{R}}

\newcommand\1{{\ensuremath {\mathds 1} }}
\renewcommand\phi{\varphi}
\newcommand{\gH}{\mathfrak{H}}
\newcommand{\gK}{\mathfrak{K}}
\newcommand{\gS}{\mathfrak{S}}

\newcommand{\cF}{\mathcal{F}}
\newcommand{\cN}{\mathcal{N}}

\newcommand{\norm}[1]{ \left\| #1 \right\|}
\newcommand{\tr}{{\rm Tr}\,}

\renewcommand{\geq}{\geqslant}
\renewcommand{\leq}{\leqslant}

\newcommand{\eps}{\varepsilon}
\usepackage{color}

\newcommand{\nn}{\nonumber}

\newcommand{\dd}{\mathrm{d}}
\newcommand{\Vext}{V_{\rm ext}}
\newcommand{\ch}{\mathrm{ch}}
\newcommand{\sh}{\mathrm{sh}}
\newcommand{\Tr}{\mathrm{Tr}}
\newcommand{\dGamma}{\mathrm{d}\Gamma}
\newcommand{\ao}{\mathfrak{a}}

\date{\today}

\begin{document}
 
\title[Optimal rate of condensation in the Gross--Pitaevskii regime]{Optimal rate of condensation for trapped bosons in the Gross--Pitaevskii regime}

\author[P.T. Nam]{Phan Th\`anh Nam}
\address{Department of Mathematics, LMU Munich, Theresienstrasse 39, 80333 Munich, and Munich Center for Quantum Science and Technology, Schellingstr. 4, 80799 Munich, Germany} 
\email{nam@math.lmu.de}

\author[M. Napi\'orkowski]{Marcin Napi\'orkowski}
\address{Department of Mathematical Methods in Physics, Faculty of Physics, University of Warsaw,  Pasteura 5, 02-093 Warszawa, Poland}
\email{marcin.napiorkowski@fuw.edu.pl} 

\author[J. Ricaud]{Julien Ricaud}
\address{Department of Mathematics, LMU Munich, Theresienstrasse 39, 80333 Munich, and Munich Center for Quantum Science and Technology, Schellingstr. 4, 80799 Munich, Germany} 
\email{ricaud@math.lmu.de}

\author[A. Triay]{Arnaud Triay}
\address{Department of Mathematics, LMU Munich, Theresienstrasse 39, 80333 Munich, and Munich Center for Quantum Science and Technology, Schellingstr. 4, 80799 Munich, Germany}  
\email{triay@math.lmu.de}

\date{\today}

\begin{abstract} We study the Bose--Einstein condensates of trapped Bose gases in the Gross--Pitaevskii regime. We show that the ground state energy and ground states of the many-body quantum system are correctly described by the Gross--Pitaevskii equation in the large particle number limit, and provide the optimal convergence rate.  Our work extends the previous results of Lieb, Seiringer and Yngvason on the leading order convergence, and of Boccato, Brennecke, Cenatiempo and Schlein on the homogeneous gas. Our method relies on the idea of `completing the square',  inspired by recent works of Brietzke, Fournais and Solovej on the Lee--Huang--Yang formula, and a general estimate for Bogoliubov quadratic Hamiltonians on Fock space.  \\


\end{abstract}

 \maketitle

 \tableofcontents

\section{Introduction}

Since the first experimental realizations of Bose--Einstein condensation in cold atomic gases in 1995 \cite{WieCor-95,Ket-95}, the rigorous understanding of the  condensation from basic laws of quantum physics has become a major problem in mathematical physics. In the present paper, we will investigate this issue for a system of trapped bosons and provide a quantitative justification of the condensation for the low-lying eigenstates. 

We consider a system of $N$ bosons in $\R^3$ described by the Hamiltonian 
\begin{equation} \label{eq:HN}
H_N= \sum\limits_{j = 1}^N (-\Delta_{x_j} + V_{\rm ext}(x_j)) +   \sum\limits_{1 \leqslant j < k \leqslant N} N^2 V(N(x_j-x_k))
\end{equation}
on the bosonic space $L^2(\R^3)^{\otimes_s N}$. Here the external potential, which satisfies $V_{\rm ext}(x) \to+\infty$ as $|x|\to +\infty$, serves to trap the system.   The interaction potential $V$ is non-negative and sufficiently smooth. The Hamiltonian $H_N$ with the core domain $C_c^\infty(\R^3)^{\otimes_s N}$ is bounded from below  and  can be extended to be a self-adjoint operator by Friedrichs' method. 

Note that the range of the interaction is of order $N^{-1}$, much smaller than the average distance of the particles (which is of order $N^{-1/3}$  as the system essentially occupies a volume of order 1). Therefore, any particle interacts with very few others, namely the system is very dilute. On the other hand, the interaction potential is very strong in its range (the strength is of order $N^2$), making the correlation of particles complicated. This so-called  {\em Gross--Pitaevskii regime} is relevant to the physical setup in \cite{WieCor-95,Ket-95}, and its mathematical analysis is both interesting and difficult. 

To the leading order, the macroscopic properties of the system are well captured by the famous Gross--Pitaevskii theory \cite{Gro-61,Pit-61}. In this theory, a quantum particle is effectively felt by the others as a hard sphere whose radius is the {\em scattering length} of the interaction potential. Recall that the scattering length $\ao$ of the potential $V$ is defined by the variational formula 
\begin{equation}\label{eq:var scat}
	8\pi \ao = \inf\left\{ \int_{\R ^3} 2 |\nabla f| ^2 + V |f| ^2 , \quad \lim_{|x|\to \infty}f(x)=1 \right\}. 
\end{equation}
When $V$ is sufficiently smooth, \eqref{eq:var scat} has a minimizer $0\le f\le 1$ and it satisfies 
\begin{equation}\label{eq:scat-intro-1}
	(-2\Delta+V)f=0. 
\end{equation}
The scattering length can then be recovered from the formula
\begin{equation}\label{eq:scat-intro-2}
	8\pi \ao = \int Vf. 
\end{equation}
By scaling,  the scattering length of $V_N=N^2V(N\cdot)$ is $\ao N^{-1}$. If we formally replace the interaction potential $V_N(x-y)$ in $H_N$ by the Dirac-delta interaction $8\pi \ao N^{-1}\delta_0(x-y)$, and insert the ansazt  of full condensation $\Psi_N= u^{\otimes N}$, then  we obtain the Gross--Pitaevskii approximation for the ground state energy per particle
\begin{align} \label{eq:eGP}
	e_{\rm GP}=  \inf_{\norm{u}_{L^2(\R^3)=1}} \int_{\R ^3} \left( |\nabla u| ^2 + V_{\rm ext} |u|^2 + 4\pi \ao |u|^4 \right).
\end{align}
It is well-known that the variational problem \eqref{eq:eGP} has a minimizer $\varphi_{\rm GP} \ge 0$. This minimizer is unique (up to a constant phase) and solves the Gross--Pitaevskii equation 
\begin{equation} \label{eq:GP-equation}
	\left(-\Delta + \Vext + 8\pi \ao \varphi_{\rm GP}^2 -\mu \right) \varphi_{\rm GP}= 0, \quad \mu \in \mathbb{R}.
\end{equation}

Note that the expectation of $H_N$ against the uncorrelated wave function  $\varphi_{\rm GP}^{\otimes N}$ gives us a formula like \eqref{eq:eGP} but with $4\pi \ao$ replaced by the larger value $(1/2)\int V$. Thus in the Gross--Pitaevskii  regime, the particle correlation due to the two-body scattering process plays a crucial role. 

\medskip

The rigorous derivation of the Gross--Pitaevskii theory  from the many-body Schr\"o\-ding\-er theory is the subject of many important works in the last two decades; see  \cite{LSY-00,LS-02,LS-06,NRS-16,BBCS-18,BBCS-19,BBCS-19b} for low-lying eigenstates, \cite{DSY-19,DS-19} for thermal equilibrium states, and  \cite{ESY-09,ESY-10,Pic-15,BOS-14,BS-19} for dynamics. In regard to the ground state energy
\[
	E_N := \inf {\rm Spec}(H_N)= \inf_{\Psi\in L^2(\R^3)^{\otimes_s N}, \norm{\Psi}_{L^2}=1} \langle \Psi, H_N \Psi\rangle, 
\]
Lieb, Seiringer and Yngvason \cite{LSY-00} proved that 
\begin{equation} \label{eq:BEC-0}
	\lim_{N\to \infty} \frac{E_N}{N} = e_{\rm GP}. 
\end{equation}
Later,  Lieb and Seiringer \cite{LS-06} proved that the ground state $\Psi_N$ of $H_N$ exhibits the complete Bose--Einstein condensation, namely 
\begin{equation} \label{eq:BEC}
	\lim_{N\to \infty} \frac{\gamma_{\Psi_N}^{(1)}} {N} = |\varphi_{\rm GP}\rangle \langle \varphi_{\rm GP}|
\end{equation}
in the trace norm. See also \cite{NRS-16} for a simplified proof of these results. Here the one-body density matrix $\gamma_{\Psi_N}^{(1)}$ of $\Psi_N$ is an operator on $L^2(\R^3)$ with kernel
\[
	\gamma_{\Psi_N}^{(1)}(x,y)= N \int_{(\R^3)^{N-1}} \Psi_N(x,x_2,...,x_N) \overline{\Psi_N(y,x_2,...,x_N)} \dd x_2 ... \dd x_N. 
\]
In particular, $\gamma_{u^{\otimes N}}^{(1)}=|u\rangle \langle u|$. Note that \eqref{eq:BEC} may hold even if
 $\Psi_N$ and $\varphi_{\rm GP}^{\otimes N}$ are not close in the usual norm topology. 

\medskip

A special case of trapped systems is the {\em homogeneous gas}, where $H_N$ acts on $L^2(\mathbb{T}^3)^{\otimes_s N}$ instead of $L^2(\R^3)^{\otimes_s N}$, with $\mathbb{T}^3$ the unit torus in three dimensions, and the external potential is ignored. For this translation-invariant system, it is easy to see that 
\[
	e_{\rm GP}=4\pi \ao, \quad \varphi_{\rm GP} =1.
\]
In this case, recently Boccato, Brennecke, Cenatiempo and Schlein \cite{BBCS-18,BBCS-19} proved that 
\begin{align} \label{eq:BEC-opt}
	E_N= N e_{\rm GP} + O(1), \quad \langle \varphi_{\rm GP}, \gamma_{\Psi_N} \varphi_{\rm GP}\rangle = N + O(1),  
\end{align} 
improving upon the leading order convergence in \cite{LS-02}. These optimal bounds are crucial inputs for the analysis of the excitation spectrum of $H_N$ in \cite{BBCS-19b}. Similar bounds were obtained earlier  for the quantum dynamics in \cite{BOS-14,BS-19}. 

\medskip

\subsection*{Main result} In the present paper we aim at providing an alternative approach to the optimal condensation \eqref{eq:BEC-opt} and extending it to inhomogeneous trapped gases. Our main result is

\begin{theorem}[Optimal condensation for trapped Bose gases] \label{thm:main} Let $0\le V_{\rm ext}\in C^1(\R^3)$ satisfy $|\nabla V_{\rm ext}(x)|^2 \le 2V_{\rm ext}(x)^{3}+C$ and $V_{\rm ext}(x)\to \infty$ as $|x|\to \infty$. Let $0\le V\in L^1(\R^3)$ be radial with compact support such that its scattering length $\ao$ is small. Let $\varphi_{\rm GP}$ be the Gross--Pitaevskii minimizer for $e_{\rm GP}$ in \eqref{eq:eGP}. Let $E_N$ be the ground state energy of the Hamiltonian $H_N$ in \eqref{eq:HN}. Then we have the energy bound 
\begin{align}  \label{eq:EN-NeGP}
	|E_N  - Ne_{\rm GP}| \le C 
\end{align}
and the operator bound 
\begin{align} \label{eq:HN>=}
	H_N \ge N e_{\rm GP} +  C^{-1}\sum_{i=1}^N (\1 - |\varphi_{\rm GP}\rangle \langle \varphi_{\rm GP}|)_{x_i} -C \quad \text{ on }\quad L^2(\R^3)^{\otimes_s N}.
\end{align}
Here $C>0$ is constant independent of $N$. 
\end{theorem}

A direct consequence of Theorem \ref{thm:main} is 
\begin{corollary} For any wave function $\Psi_N$ in $L^2(\R^3)^{\otimes_s N}$, we have 
\begin{align} \label{eq:BEC-opt-thm}
	0\le N - \langle \varphi_{\rm GP}, \gamma_{\Psi_N}^{(1)} \varphi_{\rm GP}\rangle  \le C (\langle \Psi_N, H_N \Psi_N \rangle - E_N +1). 
\end{align} 
In particular, if $\Psi_N$ is a ground state of $H_N$, then $N - \langle \varphi_{\rm GP}, \gamma_{\Psi_N}^{(1)} \varphi_{\rm GP}  \rangle \le C$. 
\end{corollary}

The precise smallness condition on the scattering length needed for our proof is 
\begin{align} \label{eq:gap-condition}
	\int_{\R^3} \left( |\nabla \varphi_{\rm GP} |^2 + \Vext |\varphi_{\rm GP}|^2 \right) + 40 \pi \ao \norm{\varphi_{\rm GP}}_{L^\infty}^2<  \inf_{\substack{u\bot \varphi_{\rm GP} \\  \norm{u}_{L^2}=1}}   \int_{\R^3} \left( |\nabla  u |^2 + \Vext |u|^2 \right). 
\end{align}
Note that when $\ao$ tends to $0$, the left and right sides of \eqref{eq:gap-condition} converge to the first and second lowest eigenvalues of $-\Delta+V_{\rm ext}$, respectively. Thanks to the spectral gap of the one-body Schr\"odinger operator, \eqref{eq:gap-condition} holds for $\ao>0$ small. 

As in \cite{BBCS-18}, the smallness condition on the interaction helps us to simplify the analysis. The case of large interaction is more difficult and left open. Heuristically, the technical smallness assumption could be removed by using the leading order result \eqref{eq:BEC-0}-\eqref{eq:BEC} as an input, plus a localization method on the number of excited particles to carefully factor out the particle correlation. This has been done  for the homogeneous case  \cite{BBCS-19}, but the extension to the inhomogeneous case would require several additional arguments. On the other hand, our proof below is direct and we do not use the leading order result \eqref{eq:BEC-0}-\eqref{eq:BEC} at all. 

Note that our proof of the lower bound \eqref{eq:scat-intro-2} can be extended to hard core potentials because we will need only the finiteness of the scattering length (more precisely we will very soon replace $V$ by $Vf$ which is a bounded measure with $\int Vf=8\pi \ao$). However, the condition $V\in L^1(\R^3)$ is important for the upper bound $E_N\le Ne_{\rm GP}+C$.

\subsection*{Proof strategy} Our proof contains two main steps. 

\begin{itemize}
	\item First, we factor out the particle correlation by simply `completing the square'. The idea is  inspired by recent works of Brietzke--Fournais--Solovej \cite{BFS-18} and Four\-nais--Solovej \cite{FS-19} on the Lee--Huang--Yang formula. This step allows us to bound $H_N$ by a quadratic Hamiltonian on Fock space, up to an energy shift.  
	\item Second, we estimate the ground state energy  of the quadratic Hamiltonian. In the homogeneous case, this step can be done by `completing the square' again, as realized already in 1947 by Bogoliubov  \cite{Bog-47}.  In the inhomogeneous case, the analysis is significantly more complicated. We will derive a sharp lower bound for general  quadratic Hamiltonians, and then apply it to the problem at hand. 
\end{itemize}

Our method is different from that of \cite{BBCS-18,BBCS-19}. For the reader's convenience, we will quickly present our proof in the homogeneous case. Then we explain further details in the inhomogeneous case. 

We use the Fock space formalism, which is recalled in Section \ref{sec:pre}. Our convention of the Fourier transform is
\[ \widehat g(p)=\int g(x)e^{-ip\cdot x} \dd x. \]

\subsubsection*{Homogeneous case} Let us focus on the main estimate \eqref{eq:HN>=}. Let $ P=|\varphi_{\rm GP} \rangle \langle \varphi_{\rm GP}|$ and let $f_N(x)=f(Nx)$ where $f$ is the scattering solution of \eqref{eq:scat-intro-1}. 
Since $V_N\ge 0$ we have 
\begin{equation} \label{eq:key-1-PPfV1-fPP}
	(\1- P\otimes P f_N) V_N (\1- f_N P\otimes P) \ge 0
\end{equation}
where  $V_N$ and $f_N$ are the multiplication operators by $N^2V(N(x-y))$ and $f(N(x-y))$ on the two-particle space. Expanding \eqref{eq:key-1-PPfV1-fPP} leads to the operator inequality
\begin{align} \label{eq:HN>=HBog}
	H_N &\ge \sum_{p\ne 0} \left( |p|^2 a_p^* a_p + \frac{1}{2} \widehat{f_NV_N}(p) a_p^* a_{-p}^* a_0 a_0 + \frac{1}{2} \widehat{f_NV_N}(p) a_0^* a_{0}^* a_p a_{-p}  \right) \nn\\
	&\quad +\frac{1}{2} \left( \int  (2f_N-f_N^2)V_N \right) a_0^* a_0^* a_0 a_0.
\end{align}
Here $a_p^*$ and $a_p$ are the creation and annihilation operators of momentum $p\in 2\pi \mathbb{Z}^3$ on Fock space. Note that the form of the `square' in \eqref{eq:key-1-PPfV1-fPP} is slightly different from that of \cite{BFS-18,FS-19} as we factor out completely the `cubic contribution' to make the analysis shorter (in [11] the cubic terms are estimated further by the Cauchy--Schwarz inequality). 

Next, recall that by an extension of Bogoliubov's method \cite[Theorem 6.3]{LS-01}, we have
\begin{align*} 
A(b_p^* b_p + b_{-p}^* b_{-p})+ B(b_p^* b_{-p}^* + b_{p}b_{-p}) \ge  - (A-\sqrt{A^2-B^2}) \frac{[b_p,b_{p}^*]+[b_{-p}, b_{-p}^*]}{2} 
\end{align*}
for all constants $A\ge B\ge 0$ and operators $b_p, b_{-p}$ on Fock space. Taking  
\[ b_p=N^{-1/2}a_0^* a_p, \quad b_p^* b_p \le a_p^* a_p, \quad [b_p,b_p^*] \le 1, \quad \forall 0\ne p \in 2\pi \mathbb{Z}^3 \]
we find that
\begin{multline} \label{eq:HN>=HBog-1}
	\sum_{p\ne 0} \left( (|p|^2-\mu) a_p^* a_p + \frac{1}{2} \widehat{f_NV_N}(p) a_p^* a_{-p}^* a_0 a_0 + \frac{1}{2} \widehat{f_NV_N}(p) a_0^* a_{0}^* a_p a_{-p}  \right)\\
	\begin{aligned}[b]
		&\ge \frac{1}{2} \sum_{p\ne 0} \left( (|p|^2-\mu) (b_p^* b_p+b_{-p}^*b_{-p}) +  N\widehat{f_NV_N}(p) b_p^* b_{-p}^* +\widehat{f_NV_N}(p) b_p b_{-p}  \right)\\
		&\ge - \frac{1}{2} \sum_{p\ne 0} \left( |p|^2-\mu - \sqrt{(|p|^2-\mu)^2 -  |N\widehat{f_NV_N}(p)|^2}  \right)
	\end{aligned}
\end{multline}
for any constant $0<\mu<4\pi^2-8\pi \ao$ (we used $N \|\widehat {f_NV_N}\|_{L^\infty}\le 8\pi \ao$). It is straightforward to see that  the right side of \eqref{eq:HN>=HBog-1} equals 
\begin{align} \label{eq:HN>=HBog-2}
	- \frac{N^2}{4}\sum_{p\ne 0} \frac{|\widehat{f_NV_N}(p)|^2}{4|p|^2} + O(1) &= - \frac{N^2}{4} \int_{\R^3} \frac{|\widehat{f_NV_N}(p)|^2}{4|p|^2} \frac{\dd p}{(2\pi)^3}+ O(1) \nn\\
	&= - \frac{N^2}{2}\int_{\R^3}  V_N f_N (1-f_N) + O(1). 
\end{align}
Here we have used Plancherel's identity and the fact that $\widehat{V_Nf_N}(p)=2|p|^2 \widehat{(1-f_N)}(p)$ which follows from the scattering equation \eqref{eq:scat-intro-1}. 

Finally, inserting \eqref{eq:HN>=HBog-1}-\eqref{eq:HN>=HBog-2} in  \eqref{eq:HN>=HBog} we conclude that, for any $\mu<4\pi^2- 8\pi \ao$, 
\begin{align} \label{eq:HN>=HBog-3}
	H_N &\ge \begin{multlined}[t]
		\mu \cN_+ -\frac{N^2}{2}\int   V_N f_N (1-f_N)\\
		+ \frac{1}{2}  (N-\cN_+)(N-\cN_+-1) \int  (2f_N-f_N^2)V_N  + O(1)
	\end{multlined}\nn\\
	& \ge  (\mu-16\pi \ao) \cN_+ + \frac{N^2}{2}\int f_N V_N  + O(1)=  (\mu-16\pi \ao) \cN_+ + 4\pi \ao N + O(1)
\end{align}
where $\cN_+:= N-a_0^*a_0$. If $\ao<\pi/6$, we can choose $\mu$ such that $16\pi \ao<\mu<4\pi^2- 8\pi \ao$ and conclude the proof of \eqref{eq:HN>=}. 

\subsubsection*{Inhomogeneous case} 
Using \eqref{eq:key-1-PPfV1-fPP} we obtain a lower bound of the form
\begin{align} \label{eq:HN>=first-bound-intro}
H_N &\ge N  \int_{\mathbb{R}^{3}} \left(|\nabla \varphi_{\rm GP} |^2 + \Vext |\varphi_{\rm GP}|^2 \right)  + \frac{N^2}{2} \int_{\R^3} \left ((((2f_N-f_N^2)V_N)*\varphi_{\rm GP}^2)\varphi_{\rm GP}^2 \right)  \nn\\
& \quad  + \inf {\rm Spec} (\mathbb{H}_{\rm Bog}) + (\mu-\mu_1)\cN_+ + O(1)
\end{align}
where 
\[ 
	\cN_+ =\dGamma(Q), \quad \mu_1=\int_{\R^3} \left( |\nabla \varphi_{\rm GP} |^2 + \Vext |\varphi_{\rm GP}|^2 \right)+ 32\pi \ao \norm{\varphi_{\rm GP}}_{L^\infty}^2,   
\]
and $\mathbb{H}_{\mathrm{Bog}}$ is an operator on the excited Fock space $\cF(Q L^2(\R^3))$ defined by
\begin{align} \label{eq:H-K-intro}
	\mathbb{H}_{\mathrm{Bog}} &= \dGamma(H) + \frac{1}{2} \iint K(x,y) (a^*_x a^*_y +  a_x a_y)   \, \dd x \, \dd  y,\nn\\
	H &=Q(-\Delta+V_{\rm ext}-\mu)Q,\quad Q=\1-P\\
	\intertext{and}
	K(x,y) &=(Q\otimes Q \widetilde K (\cdot, \cdot) )(x,y), \quad \widetilde K(x,y)= \varphi_{\rm GP}(x) \varphi_{\rm GP}(y) (NV_Nf_N)(x-y). \nn
\end{align}

Note that \eqref{eq:gap-condition} allows us to choose $\mu>\mu_1$ such that $H> \|K\|_{\rm op}$, where $K$ is the operator with kernel $K(x,y)$. Therefore, in  principle, the quadratic Hamiltonian $\mathbb{H}_{\mathrm{Bog}}$ can be diagonalized by a Bogoliubov transformation; see  \cite{GS-13,BB-15,NNS-16,Der-17} for recent results. However, extracting an explicit lower bound is not straightforward. We will prove  the following general lower bound, which is of independent interest, 
\begin{equation} \label{eq:HBog>=abs-intro}
	\mathbb{H}_{\rm Bog} \ge - \frac{1}{4} \tr\left(H^{-1} K^2 \right)  - C  \norm{K}_{\rm op} \Tr(H^{-2} K^{2}) 
\end{equation}
(see Theorem \ref{thm:general-bound-HBog}).  The simpler general lower bound 
\begin{equation} \label{eq:Bog-half-lower-bound}
\mathbb{H}_{\rm Bog} \ge - \frac{1}{2} \tr\left(H^{-1} K^2 \right) 
\end{equation}
is well-known; see \cite[Theorem
5.4]{BD-07},  \cite[Theorem 2]{NNS-16} and \cite[Theorem 3.23]{Der-17}. The significance of \eqref{eq:quad-Hamil-lwb-general} is that we get the optimal constant $(-1/4)$ for the main term which is crucial for our application. 

It remains to evaluate the right side of  \eqref{eq:HBog>=abs-intro} for $H$ and $K$ in \eqref{eq:H-K-intro}. For a heuristic calculation, let us replace $H$ and $K$ by $-\Delta$ and $\widetilde K$, respectively. We can write
\[
	\Tr\left((-\Delta)^{-1} \widetilde{K}^2\right) = N^2 \tr\left(\varphi_{\rm GP}(x) \widehat{(f_NV_N)} (p) \varphi_{\rm GP}(x) p^{-2} \varphi_{\rm GP}(x) \widehat{(f_NV_N)} (p) \varphi_{\rm GP}(x) \right)
\]
where $\varphi_{\rm GP}(x)$ and $v(p)$ are multiplication operators in the position and momentum spaces. If  we could commute $\varphi_{\rm GP}(x)$ and $p^{-2}$, then the above trace would become
\begin{multline*}
	N^2 \tr\left(\varphi_{\rm GP} (x)\widehat{(f_NV_N)} (p)  \varphi_{\rm GP}^2(x) |p|^{-2} \widehat{(f_NV_N)} (p) \varphi_{\rm GP}(x) \right) \\
	\begin{aligned}
		&=2 N^2 \tr\left(\varphi_{\rm GP}^2 (x)\widehat{(f_NV_N)} (p)  \varphi_{\rm GP}^2(x) \widehat{(1-f_N)} (p) \right) \\
		&= 2N^2 \iint \varphi_{\rm GP}^2 (x) (V_Nf_N(1-f_N))(x-y) \varphi_{\rm GP} ^2 (y) \dd x \dd y.
	\end{aligned}
\end{multline*}
Here we have used $\widehat{V_Nf_N}(p)=2|p|^2 \widehat{(1-f_N)}(p)$ thanks to the  scattering equation \eqref{eq:scat-intro-1} and the equality between the Hilbert--Schmidt norm of operators and the $L^2$-norm of operator kernels. This heuristic calculation can be made rigorous by using the Kato--Seiler--Simon inequality \cite[Theorem 4.1]{Sim-05} to control several commutators. We can also bound $\Tr(H^{-2} K^{2})$ by $O(1)$ in the same way. 

In summary, we obtain from \eqref{eq:HN>=first-bound-intro} and \eqref{eq:HBog>=abs-intro} that 
\begin{align*} 
	H_N &\ge \begin{multlined}[t]
		N \int_{\mathbb{R}^{3}} \left(|\nabla \varphi_{\rm GP} |^2 + \Vext |\varphi_{\rm GP}|^2 \right)  + \frac{N^2}{2} \int_{\R^3} \left(\left(\left(\left(2f_N-f_N^2\right)V_N\right)*\varphi_{\rm GP}^2\right)\varphi_{\rm GP}^2 \right)\\
		+ \frac{N^2}{2} \int \left(\left(V_Nf_N\left(1-f_N\right)\right)* \varphi_{\rm GP} ^2\right) \varphi_{\rm GP}^2  + (\mu-\mu_1)\cN_++ O(1)
	\end{multlined}\\
	&= \begin{multlined}[t]
		N \int_{\mathbb{R}^{3}} \left(|\nabla \varphi_{\rm GP} |^2 + \Vext |\varphi_{\rm GP}|^2 \right) + \frac{N^2}{2}\int \left(\left(V_Nf_N\right)*\varphi_{\rm GP}^2\right) \varphi_{\rm GP}^2 \\
		+ (\mu-\mu_1)\cN_+ + O(1)
	\end{multlined}\\
	&= N \int_{\mathbb{R}^{3}} \left(|\nabla \varphi_{\rm GP} |^2 + \Vext |\varphi_{\rm GP}|^2 \right) + 4\pi \ao \int \varphi_{\rm GP}^4 + (\mu-\mu_1)\cN_+ + O(1). 
\end{align*}
Here we have used $NV_Nf_N\approx 8\pi \ao \delta_0$ in the last estimate. Thus \eqref{eq:HN>=} holds true.

The energy upper bound $E_N\le Ne_{\rm GP}+O(1)$ is a separate issue, which is conceptually easier. It is known that in the Fock space setting, a good trial state is of the form 
\[ W(\sqrt{N}\varphi_{\rm GP}) \Gamma' W(\sqrt{N}\varphi_{\rm GP})^* \]
where $W(g)=e^{a(g)-a^*(g)}$ is the Weyl operator and $\Gamma'$ is an appropriate quasi-free state; see Benedikter--Porta--Schlein \cite[Appendix A]{BPS-16} (similar constructions in the homogeneous case can be found in \cite{ESY-08,NRS-18}). This construction can be adapted to the $N$-particle Hilbert space, using the unitary operator $U_N$ introduced by Lewin--Nam--Serfaty--Solovej \cite{LNSS-15} instead of the Weyl operator and a modified version of $\Gamma'$.

\bigskip

\noindent{\bf Organization of the paper.} In Section \ref{sec:pre} we recall some standard facts on the scattering length, Gross--Pitaevskii theory, Fock space formalism and quasi-free states. Then we prove the operator lower bound \eqref{eq:HN>=} in Section \ref{sec:inhom} and the energy upper bound in Section \ref{sec:upp}.

\bigskip

\noindent{\bf Acknowledgments.} We thank S\o ren Fournais, Jan Philip Solovej and Benjamin Schlein for helpful discussions. We thank the referees for constructive comments and suggestions. We received funding from the Deutsche Forschungsgemeinschaft (DFG, German Research Foundation) under Germany's Excellence Strategy (EXC-2111-390814868), and from the National Science Centre (NCN project Nr. 2016/21/D/ST1/02430).

\section{Preliminaries} \label{sec:pre}

\subsection{Scattering length} We recall some basic properties of the scattering length from \cite[Appendix C]{LSSY-05}. Under our assumptions on the potential $V$, the scattering problem \eqref{eq:var scat} has a unique minimizer $f$. The minimizer is radially symmetric, $0\le f\le 1$ and 
\begin{equation} \label{eq:scattering-int}
	 (-2\Delta+V(x))f(x)=0, \quad 8\pi \ao = \int V f.
\end{equation}
Moreover, the function $\omega=1-f$ vanishes at infinity, more precisely
$$
0\le \omega(x) \le \frac{C}{|x|+1}, \quad \forall x\in \R^3. 
$$

By scaling, the function $f_N(x)=f(Nx)$ solves the scattering problem for $V_N(x)=N^2V(Nx)$, namely 
\begin{equation} \label{eq:scatering-eq-N}
	(-2\Delta+V_N(x))f_N(x)=0, \quad  \frac{8\pi \ao}{N} = \int V_N f_N. 
\end{equation}
Thus the function $\omega_N=1-f_N$ vanishes at infinity, 
\begin{equation}\label{eq:bounds_on_w}
	0\le w_N(x) \leq \frac{C}{|Nx|+1}, 
\end{equation}
 and satisfies 
\begin{equation} \label{eq:scatering-eq-N-omega} 
	- 2 \Delta \omega_N = V_Nf_N  \quad \text{ in } \R^3.
\end{equation}

\subsection{Gross--Pitaevskii theory}  Under our assumption on the external potential $V_{\rm ext}$, the minimization problem \eqref{eq:eGP} has a minimizer $\varphi_{\rm GP} \ge 0$ under the constraint 
\[ \varphi_{\rm GP} \in H^1(\R^3), \quad \int |\varphi_{\rm GP}|^2=1,  \quad \int V_{\rm ext} |\varphi_{\rm GP}|^2 <\infty. \]
Moreover, the minimizer is unique (up to a constant phase) and  satisfies the Euler--Lagrange equation \eqref{eq:GP-equation}; see \cite[Appendix A]{LSY-00} for details. 

From \eqref{eq:GP-equation} and the fact $H^1(\R^3)\subset L^6(\R^3)$, we find that $(-\Delta+V_{\rm ext})\varphi_{\rm GP} \in L^2(\R^3)$. Under the extra condition $|\nabla V_{\rm ext}|^2 \le 2V_{\rm ext}^3 + C$ we can show that $\Delta\varphi_{\rm GP}\in L^2(\R^3)$ as follows. Replacing $V_{\rm ext}$ by $V_{\rm ext}+1$ if necessary, we can assume that $V_{\rm ext} \ge 1$.  By the IMS formula
\begin{align} \label{eq:IMS-formula}
	g^2(x) (-\Delta)+ (-\Delta) g^2(x) = 2g(x) (-\Delta)g(x) -  2 |\nabla g(x)|^2, \quad \forall 0\le g\in H^1,
\end{align}
we can write
 \begin{align*}
(-\Delta+V_{\rm ext})^2 &= \Delta^2 + V_{\rm ext}^2  + V_{\rm ext} (-\Delta) + (-\Delta) V_{\rm ext} \\
&= \Delta^2 + V_{\rm ext}^2 + 2\sqrt{V_{\rm ext}} (-\Delta)  \sqrt{V_{\rm ext}}  - 2 | \nabla \sqrt{V_{\rm ext}}|^2. 
\end{align*}
The condition $|\nabla V_{\rm ext}|^2 \le 2V_{\rm ext}^3 + C$  ensures that 
$$
2 | \nabla \sqrt{V_{\rm ext}}|^2 =  \frac{|\nabla V_{\rm ext}|^2}{2V_{\rm ext}} \le V_{\rm ext}^2 +C.
$$
Therefore, we conclude that 
\begin{equation} \label{eq:H^2>Delta^2}
	(-\Delta+V_{\rm ext})^2  \ge \Delta^2 -C \quad \text{ on }L^2(\R^3).
\end{equation}
Consequently, from $(-\Delta+V_{\rm ext})\varphi_{\rm GP} \in L^2(\R^3)$ we deduce that $\Delta \varphi_{\rm GP}\in L^2(\R^3)$. In summary, we have $\varphi_{\rm GP} \in H^2(\R^3)\subset L^\infty(\R^3)$. 

\begin{remark}
We can replace our assumption $|\nabla V_{\rm ext}|^2 \le 2V_{\rm ext}^3 + C$ by \eqref{eq:H^2>Delta^2},  or slightly more general
$$
	(-\Delta+V_{\rm ext})^2  \ge C^{-1}\Delta^2 -C \quad \text{ on }L^2(\R^3).
$$
This kind of conditions is natural to ensure that the operator domain $D(-\Delta+V_{\rm ext})$ is a subspace of $H^2(\R^3)$. In general, if $0\le V_{\rm ext}\in L^2_{\rm loc}(\R^d)$, then $-\Delta+V_{\rm ext}$ is essentially self-adjoint with core domain $C_c^\infty(\R^d)$ by Kato's theorem \cite[Theorem X.28]{Reed-Simon-Vol2}, but  $D(-\Delta+V_{\rm ext})$ may be different from $D(-\Delta)\cap D(V_{\rm ext})$ \cite[Theorem X.32]{Reed-Simon-Vol2}. See  \cite{Glimm-Jaffe,Davies} for further discussions in this direction. 
\end{remark}

\subsection{Fock space formalism} Let $\gH=L^2(\R^d)$ (or a closed subspace of $L^2(\R^d)$)  be the Hilbert space of one particle. The bosonic Fock space is defined by 
\[ \cF(\gH)= \bigoplus_{n=0}^\infty \gH^{\otimes_s n} \]
where the number of particles can vary. For any $g\in \gH$, we can define the creation and annihilation operators $a^*(f)$, $a(f)$ on  Fock space by 
\begin{align*}
	(a^* (g) \Psi )(x_1,\dots,x_{n+1})&= \frac{1}{\sqrt{n+1}} \sum_{j=1}^{n+1} g(x_j)\Psi(x_1,\dots,x_{j-1},x_{j+1},\dots, x_{n+1}), \\
	(a(g) \Psi )(x_1,\dots,x_{n-1}) &= \sqrt{n} \int_{\R^d} \overline{g(x_n)}\Psi(x_1,\dots,x_n) \dd x_n, \quad \forall \Psi \in \gH^n,\, \forall n. 
\end{align*}
These operators satisfy the canonical commutation relations
\[ 
	[a(g_1),a(g_2)]=[a^*(g_1),a^*(g_2)]=0,\quad [a(g_1), a^* (g_2)]= \langle g_1, g_2 \rangle, \quad \forall g_1,g_2 \in \gH.
\]
We may also define 
the operator-valued distributions $a_x^*$ and $a_x$, with $x\in \R^d$, by 
$$
a_x^*= \sum_{n=1}^\infty \overline{f_n(x)} a^*(f_n), \quad a_x= \sum_{n=1}^\infty f_n(x) a(f_n)
$$
where $\{f_n\}_{n=1}^\infty$ is an orthonormal basis of $\gH$ (the definition is independent of the choice of the basis). Equivalently, we have
\[
	a^*(g)=\int_{\R^d}   g(x) a_x^* \dd x, \quad a(g)=\int_{\R^d}  \overline{g(x)} a_x \dd x, \quad \forall g\in \gH.
\]
The canonical commutation relations can be rewritten as  
\[ [a^*_x,a^*_y]=[a_x,a_y]=0, \quad [a_x,a^*_y]=\delta(x-y), \quad \forall x,y\in \R^d. \]

These creation and annihilation operators can be used to express several important observables. For example, the particle number operator can be written as 
\[
	\cN := \bigoplus_{n=0}^\infty n \1_{\gH^{\otimes_s n}}  = \sum_n a^*(u_n) a(u_n) =\int_{\R^d} a_x^* a_x \dd x. 
\]
Here $\{u_n\}$ is any orthonormal basis for $\gH$. More generally, for any one-body self-adjoint operator $A$ we have
\[
	\dd \Gamma(A):= \bigoplus_{n=0}^\infty \left(  \sum_{i=1}^n A_{x_i} \right) = \sum_{m,n} \langle u_m, A u_n\rangle a_m^* a_n = \int_{\R^d} a^*_x A_x a_x \dd x.
\]
For $H_N$ in \eqref{eq:HN}, we can write 
\begin{align} \label{eq:2nd-Q}
	H_N &= \dd \Gamma (-\Delta + V_{\rm ext}) +  \frac{1}{2} \sum_{m,n,p,q} \langle u_m\otimes u_n, V_N u_p\otimes u_q\rangle a^*(u_m)a^*(u_n)a(u_p) a(u_q)\nn\\
	&= \dd \Gamma (-\Delta + V_{\rm ext})  + \frac{1}{2}\int_{\R^d}\int_{\R^d} V_N(x-y) a_x^* a_y^* a_x a_y \dd x \dd y.
\end{align}

\subsection{Quasi-free states} Let $\Gamma$ be a (mixed) state  on Fock space with finite particle number expectation, namely $\langle \cN \rangle_\Gamma = \Tr(\cN \Gamma)<\infty$. We call $\Gamma$ a {\em quasi-free state} if it satisfies Wick's Theorem: 
\begin{align*}
	&\langle a^{\#}(f_{1}) a^{\#}(f_{2}) \cdots a^{\#}(f_{2n})  \rangle_{\Gamma} = \sum_{\sigma} \prod_{j=1}^n \langle  a^{\#}(f_{\sigma(2j-1)}) a^{\#}(f_{\sigma(2j)}) \rangle_{\Gamma} , \\
	&\langle \langle a^{\#}(f_{1}) a^{\#}(f_{2}) \cdots a^{\#}(f_{2n-1})  \rangle_{\Gamma} = 0,  \quad \forall f_1,..,f_n\in \gH, \forall n\in \mathbb{N}. 
\end{align*}
Here $a^{\#}$ is either the creation or annihilation operator and the sum is taken over all permutations $\sigma$ satisfying $\sigma(2j-1)<\min\{\sigma(2j),\sigma(2j+1) \}$ for all $j$.  

By the definition any quasi-free state is determined uniquely  by its one-body density matrices $(\gamma_\Gamma,\alpha_\Gamma)$, where $\gamma_\Gamma: \gH\to \gH$ and $\alpha_\Gamma:\gH\to \gH^* \equiv \overline{\gH}$ defined by
\[
	\left\langle {g_1,{\gamma _\Gamma }g_2} \right\rangle  = \left\langle {{a^*}(g_2)a(g_1)} \right\rangle_\Gamma,\quad \left\langle {\overline{g_1}, \alpha _\Gamma {g_2} } \right\rangle  = \left\langle  {a^*(g_2)a^*(g_1)} \right\rangle_\Gamma, \quad \forall g_1,g_2 \in \gH.
\]

It is well-known (see e.g. \cite[Theorem 3.2]{Nam-11}) that any given operators $(\gamma,\alpha)$,  with $\gamma: \gH\to \gH$ and $\alpha: {\gH} \to  \gH^*  \equiv \overline{\gH}$,  are the one-body density matrices of a (mixed) quasi-free state with finite particle number expectation  if and only if 
\begin{align} \label{eq:1-pdm-quasi}
	\gamma\ge 0, \quad \Tr \gamma <\infty, \quad \overline{\alpha}=\alpha^*, \quad \begin{pmatrix}
	\gamma &  \alpha^* \\
	\alpha & 1 + \overline{\gamma} 
	\end{pmatrix} \ge 0 \quad \text{on } \gH \oplus \gH^*. 
\end{align}
Here we write $\overline{A}=JAJ$ for short, with $J$ the complex conjugation,  namely $\overline{A} g = \overline{(A \overline{g})}$.  

The reader may think of the quasi-free states as ``Gaussian quantum states". In particular, the contribution of sectors  with high particle numbers decays very fast. In fact, if $\Gamma$ is a quasi-free state, then 
\begin{align} \label{eq:fluc-N}
\langle \cN^\ell \rangle_\Gamma \le C_\ell (1+ \langle \cN\rangle_\Gamma )^\ell, \quad \forall \ell\ge 1.
\end{align}
Here the constant $C_\ell$ depends only on $\ell$  (see \cite[Lemma 5]{NN-17}).

\section{Lower bound} \label{sec:inhom}

In this section, we will prove the operator lower bound \eqref{eq:HN>=}. 

\begin{lemma}[Lower bound] 
	\label{lem:lwb} Let $V_{\rm ext}$ and $V$ be as in Theorem~\ref{thm:main}, where the scattering length $\ao$ of $V$ is small so that \eqref{eq:gap-condition} holds true. Then
\[
	H_N \ge N e_{\rm GP} +  C^{-1}\sum_{i=1}^N Q_{x_i}  -C \quad \text{ on }\quad L^2(\R^3)^{\otimes_s N}
\]
with $Q= \1 - |\varphi_{\rm GP}\rangle \langle \varphi_{\rm GP}|$. The constant $C>0$ is independent of $N$. 
\end{lemma}

\subsection{Reduction to quadratic Hamiltonian} Denote  
\begin{align}
	\mu_1 &:= \int_{\R^3} \left( |\nabla \varphi_{\rm GP} |^2 + \Vext |\varphi_{\rm GP}|^2 \right)+ 32\pi \ao \norm{\varphi_{\rm GP}}_{L^\infty}^2, \label{eq:def-mu1}\\
	\mu_2 &:=  \inf_{\substack{u\bot \varphi_{\rm GP} \\  \norm{u}_{L^2}=1}}   \int_{\R^3} \left( |\nabla  u |^2 + \Vext |u|^2 \right) - 8\pi \ao \norm{\varphi_{\rm GP}}_{L^\infty}^2. \label{eq:def-mu2}
\end{align}
Note that  $\mu_1<\mu_2$ thanks to \eqref{eq:gap-condition}. Our starting point is 

\begin{lemma} \label{lem:red-qua} Let $\mu_1<\mu<\mu_2$. Then under the notations in \eqref{eq:H-K-intro} we have
\begin{align} \label{eq:HN>=first-bound}
	H_N &\ge N  \int_{\mathbb{R}^{3}} \left(|\nabla \varphi_{\rm GP} |^2 + \Vext |\varphi_{\rm GP}|^2 \right)  + \frac{N^2}{2} \int_{\R^3} \left ((((2f_N-f_N^2)V_N)*\varphi_{\rm GP}^2)\varphi_{\rm GP}^2 \right)  \nn\\
	& \quad +  (\mu- \mu_1) \cN_+ +  \inf {\rm Spec} (\mathbb{H}_{\rm Bog})- C.
\end{align}
\end{lemma}

\begin{proof} We write $\varphi=\varphi_{\rm GP}$ for short.  Denote $P= |\varphi \rangle \langle \varphi|=\1-Q$. Let $f_N$ be the scattering solution of $V_N(x)=N^2V(Nx)$ as in \eqref{eq:scatering-eq-N}. 
 Expanding the operator inequality \eqref{eq:key-1-PPfV1-fPP}, we obtain 
\begin{align} \label{eq:VN>=inhomogeneous}
	V_N &\ge P\otimes P \left(2f_N- f_N^2\right) V_N P \otimes P + \left(P\otimes P f_N V_N Q \otimes Q + h. c. \right) \nn \\
	&\quad + \left( P\otimes P f_N V_N P\otimes Q + P\otimes P f_N V_N Q\otimes P + h. c. \right).
\end{align}
Let $\{\varphi_n\}_{n=0}^\infty$ be an orthonormal basis for $L^2(\R^3)$ with $\varphi_0=\varphi$ and denote  $a_n:=a(\varphi_n)$. 
From \eqref{eq:VN>=inhomogeneous} we have the operator inequality in $L^2(\R^3)^{\otimes_s N}$:
\begin{equation} \label{eq:HN>H0H1H2}
	H_N\ge \mathcal{H}_0 + \mathcal{H}_1 + \mathcal{H}_2
\end{equation}
where
\begin{align*}
	\mathcal{H}_0 & = \int_{\mathbb{R}^{3}} (|\nabla \varphi |^2 + \Vext \varphi^2)  + \frac{1}{2} \int_{\mathbb{R}^{3}} (((2f_N-f_N^2)V_N)*\varphi^2)\varphi^2 a_0^* a_0^* a_0 a_0,\\   	
	\mathcal{H}_1 & = a^*(Q(-\Delta + \Vext) \varphi) a_0 +  a^*(Q( (f_N V_N) \ast\varphi^2) \varphi)) a_0^* a_0 a_0  + {\rm h.c.},\\
	\mathcal{H}_2 & = \frac{1}{2}\sum_{m,n\ge 1} \left( \langle \varphi_m, (-\Delta + \Vext) \varphi_n \rangle a_m^* a_n + N^{-1}\langle \varphi_m\otimes \varphi_n, K \rangle a_m^* a_n^* a_0 a_0 + {\rm h.c.}\right). 
\end{align*}

\subsubsection*{Analysis of $\mathcal{H}_0$.} Using \eqref{eq:HN>=HBog-3} again we have 
\begin{align*}
	\mathcal{H}_0 & =\begin{multlined}[t]
		\int_{\mathbb{R}^{3}} (|\nabla \varphi |^2 + \Vext \varphi^2) (N-\cN_+)\\
		+ \frac{1}{2} \int_{\mathbb{R}^{3}} (((2f_N-f_N^2)V_N)*\varphi^2)\varphi^2  (N-\cN_+)(N-\cN_+-1)
	\end{multlined}\\
	&\ge\begin{multlined}[t]
		N  \int_{\mathbb{R}^{3}} \left(|\nabla \varphi |^2 + \Vext \varphi^2 \right)  + \frac{N^2-N}{2} \int_{\R^3} \left ((((2f_N-f_N^2)V_N)*\varphi^2)\varphi^2 \right) \\
		- \left(  \int_{\mathbb{R}^{3}} (|\nabla \varphi |^2 + \Vext \varphi^2) + N \int_{\R^3} \left ((((2f_N-f_N^2)V_N)*\varphi^2)\varphi^2 \right) \right) \cN_+.
	\end{multlined}
\end{align*}
Then using 
\[
0\le \int_{\mathbb{R}^{3}} N(((2f_N-f_N^2)V_N)*\varphi^2)\varphi^2 \le 2N \norm{f_NV_N}_{L^1} \norm{\varphi}_{L^4}^4 \le 16 \pi \ao \norm{\varphi}_{L^\infty}^2
\]
and the definition of $\mu_1$ in \eqref{eq:def-mu1}, we obtain 
\begin{align} \label{eq:H0>=}
\mathcal{H}_0 &\ge N  \int_{\mathbb{R}^{3}} \left(|\nabla \varphi |^2 + \Vext \varphi^2 \right)  + \frac{N^2}{2} \int_{\R^3} \left ((((2f_N-f_N^2)V_N)*\varphi^2)\varphi^2 \right) \nn \\
& \quad + (16\pi \ao \norm{\varphi}_{L^\infty}^2 -\mu_1 ) \cN_+  -  C. 
\end{align}

\subsubsection*{Analysis of $\mathcal{H}_1$.} We have
\begin{align} \label{eq:cH1-0}
	\mathcal{H}_1 & = a^*(Q(-\Delta + \Vext) \varphi) a_0 +  a^*(Q( (f_N V_N) \ast\varphi^2) \varphi)) (N-\cN_+) a_0  + {\rm h.c.} \nn\\
	&= a^* (Q (-\Delta + \Vext+ (N f_N V_N) \ast\varphi^2 )\varphi) a_0 -  a^*(Q( (f_N V_N) \ast\varphi^2) \varphi)) \cN_+ a_0 + {\rm h.c.}\nn\\
	&= a^* (Q ((N f_N V_N) \ast\varphi^2 - 8\pi \ao \varphi^2 )\varphi) a_0  - a^*(Q( (f_N V_N) \ast\varphi^2) \varphi)) \cN_+ a_0 + {\rm h.c.}
\end{align}
Here  in the last equality we have used the Gross--Pitaevskii equation \eqref{eq:GP-equation}. For the first term on the right side of \eqref{eq:cH1-0}, denoting 
\[ g=Q ((N f_N V_N) \ast\varphi^2 - 8\pi \ao \varphi^2 )\varphi \]
we have 
\begin{align*}
	\norm{g}_{L^2} & \le  \norm{(N f_N V_N) \ast\varphi^2 - 8\pi \ao \varphi^2}_{L^\infty} \\
	&\le \norm{N (f_N V_N*\varphi^2)^{\wedge}  - 8\pi \ao \widehat{\varphi^2}}_{L^1}= \int_{\R^3} | \widehat {fV}(p/N) - \widehat {fV}(0) | |\widehat{\varphi^2}(p)| \dd p \\
	&\le \frac{\|\nabla_p  \widehat {fV}\|_{L^\infty}}{N} \int_{\R^3}  |p| |\widehat{\varphi^2}(p)|  \dd p \le \frac{1}{N}  \norm{|x| fV}_{L^1} \norm{\varphi^2}_{H^{1/2}}\le \frac{C}{N}.
\end{align*}
Therefore, by  the Cauchy--Schwarz inequality
\[
	\pm ( a^*(g) a_0 + a_0^* a(g)) \le N a^*(g) a(g) + N^{-1} a_0^* a_0  \le N \norm{g}_{L^2}^2 \cN_+ + 1 \le C.
\]
For the second term on the right side of \eqref{eq:cH1-0}, we use 
\[
	\norm{Q ((f_N V_N) \ast\varphi^2)\varphi}_{L^2} \le \norm{(f_N V_N) \ast \varphi^2}_{L^\infty} \le \norm{f_N V_N}_{L^1} \norm{\varphi^2}_{L^\infty} \le \frac{8\pi \ao \norm{\varphi}_{L^\infty}^2 }{N},
\]
and the Cauchy--Schwarz inequality
\begin{multline} \label{eq:cH1-0-second}
	\pm \left( a^*(Q( (f_N V_N) \ast\varphi^2) \varphi)) \cN_+ a_0 + {\rm h.c.} \right)\\
	\begin{aligned}[b]
 		&\le\eps^{-1} a^*(Q( (f_N V_N) \ast\varphi^2) \varphi)) a(Q( (f_N V_N) \ast\varphi^2) \varphi)) + \eps a_0^* \cN_+^2 a_0\\
		&\le \eps^{-1} \norm{Q(f_N V_N) \ast\varphi^2) \varphi}_{L^2}^2 \cN_+  + \eps N \cN_+^2\\
		&\le \eps^{-1} \left( \frac{8\pi \ao \norm{\varphi}_{L^\infty}^2 }{N} \right)^2 \cN_+  + \eps N^2 \cN_+.
	\end{aligned}
\end{multline}
Optimizing over $\eps>0$ we can replace the right side of \eqref{eq:cH1-0-second} by $16\pi \ao \norm{\varphi}_{L^\infty}^2  \cN_+$. Thus 
\begin{equation} \label{eq:H1>=}
	\mathcal{H}_1 \ge - 16\pi \ao \norm{\varphi}_{L^\infty}^2  \cN_+ - C. 
\end{equation}

\subsubsection*{Analysis of $\mathcal{H}_2$} We will prove that 
\begin{equation} \label{eq:H2>=} 
	\mathcal{H}_2 -\mu \cN_+ \ge \inf {\rm Spec} (\mathbb{H}_{\rm Bog}). 
\end{equation}
The main difficulty in \eqref{eq:H2>=} is to remove the factors $N^{-1}a^*_0a^*_0$ and $N^{-1}a_0a_0$ in $\mathcal{H}_2$.

Recall the operators $H,K$ defined in \eqref{eq:H-K-intro}. For any (mixed) state $\Gamma$ on $L^2(\R^3)^{\otimes_s N}$, we can write
 \begin{equation*}
	\braket{\mathcal{H}_2-\mu \cN_+}_{\Gamma} = \tr (H \gamma) + \Re \tr (K\alpha)
\end{equation*}
where $K: \gH_+^* \to \gH_+$ is the operator with kernel $K(x,y)$ and $\gamma : \gH_+ \to  \gH_+$, $\alpha :  \gH_+ \to  \gH_+^* \equiv \overline{\gH_+}$ are operators defined by 
\begin{equation*}
	\braket{g_1,\gamma g_2} = \braket{ a^*(g_2) a(g_1)}_\Gamma, \quad \braket{\overline{g_1}, \alpha g_2} = N^{-1}\braket{ a^*(g_2) a^*(g_1) a_0a_0}_\Gamma, \quad \forall g_1,g_2\in \gH_+.
\end{equation*}
Then we have $\gamma\ge 0$, $\Tr \gamma= \langle \cN_+\rangle_{\Gamma}<\infty$ and $\alpha^*=\overline{\alpha}$. Moreover, for all $g_1,g_2\in \gH_+$, by the Cauchy--Schwarz inequality, we have 
\begin{align*}
	\pm(a^*(g_1)a^*(g_2) a_0 a_0 + {\rm h.c.}) &\le a^*(g_1) a_0 (a^*(g_1)a_0)^*+  (a^*(g_2)a_0)^* a^*(g_2)a_0 \\
	&= a^* (g_1) a(g_1)  (N-\cN_+ +1) +  a(g_2) a^*(g_2) (N-\cN_+)\\
	&\le N a^* (g_1)a(g_1) + N a(g_2) a^*(g_2). 
\end{align*}
Here we have used $a^*(g_1)a(g_1)(\cN_+-1)\ge 0$. Consequently, for all $g_1,g_2\in \gH_+$, 
\begin{multline*}
	\left \langle
	\begin{pmatrix}
		g_1 \\
		\overline{g_2} 
	\end{pmatrix},
	\begin{pmatrix}
		\gamma &  \alpha^* \\
		\alpha & 1 + \overline{\gamma} 
	\end{pmatrix}
	\begin{pmatrix}
		g_1 \\
		\overline{g_2}
	\end{pmatrix}
	\right\rangle_{\gH_+\oplus \gH_+^*} = \langle g_1, \gamma g_1 \rangle + \langle g_2, (1+\gamma) g_2\rangle + \langle  \overline{g_2}, \alpha {g_1}  \rangle + \langle g_1, \alpha^* \overline{g_2}\rangle \\
	=\left\langle a^*(g_1)a(g_1) + a(g_2) a^*(g_2) + N^{-1} a^*(g_2) a^*(g_1) a_0 a_0 + N^{-1} a_0^* a_0^* a(g_2)a(g_1)  \right\rangle_{\Gamma} \ge 0. 
\end{multline*}
Thus $(\gamma,\alpha)$ satisfies the conditions in \eqref{eq:1-pdm-quasi}. Hence, there exists a mixed quasi-free state $\Gamma'$ on Fock space $\cF(\gH_+)$ such that $(\gamma,\alpha)$ are its one-body density matrices. Therefore, 
\[
	\braket{\mathcal{H}_2-\mu \cN_+}_{\Gamma} = \tr (H \gamma) + \Re \tr (K\alpha) = \tr \left(  \mathbb{H}_{\mathrm{Bog}} \Gamma' \right) \ge \inf {\rm Spec} (\mathbb{H}_{\rm Bog}).
\]
Thus \eqref{eq:H2>=} holds true. 

\subsubsection*{Conclusion} Inserting \eqref{eq:H0>=}, \eqref{eq:H1>=} and \eqref{eq:H2>=} in \eqref{eq:HN>H0H1H2} we obtain \eqref{eq:HN>=first-bound}.
\end{proof}

\subsection{A general bound for quadratic Hamiltonians} Now we prove a general lower bound on quadratic Hamiltonians on Fock space, which is of independent interest.   

\begin{theorem} [Lower bound for quadratic Hamiltonians] \label{thm:general-bound-HBog} Let $\gK$ be a closed subspace of $L^2(\R^d)$. Let $K$ be a self-adjoint bounded operator on $\gK$ with real-valued symmetric kernel $K(x,y)=K(y,x)$. Let $H>0$ be a self-adjoint operator on $\gK$ such that $H$ has compact resolvent with real-valued eigenfunctions, $H^{-1/2}K$ is a Hilbert-Schmidt operator, and $H \ge (1+\eps) \|K\|_{\rm op}$ for a constant $\eps>0$.  Then 
\begin{equation} \label{eq:quad-Hamil-lwb-general}
	\dGamma(H)  + \frac{1}{2} \iint K(x,y) (a_x^* a_y^* + a_x a_y) \dd x \dd y  \ge - \frac{1}{4} \tr\left(H^{-1} K^2 \right)  - C_\eps \norm{K}_{\rm op} \Tr(H^{-2} K^{2})
\end{equation}
on the Fock space $\cF(\gK)$. Here the constant $C_\eps>0$ depends only on $\eps$. 
\end{theorem}

\begin{remark} The constant $-1/4$ is optimal. In fact, if $H$ and $K$ commute, then the ground state energy of the quadratic Hamiltonian is
\[
 \frac{1}{2} \Tr \left(\sqrt{H^2-K^2}-H\right) =  - \frac{1}{2} \Tr \left( \frac{K^2}{H+ \sqrt{H^2-K^2}}\right)
\]
 which is close to $-(1/4) \Tr(H^{-1}K^2)$ when $H$ is significantly bigger than $K$. 
\end{remark}

\begin{proof}[Proof of Theorem ~\ref{thm:general-bound-HBog}] First, let us assume that $K$ is trace class. Then following the analysis of Grech--Seiringer \cite[Section 4]{GS-13}, we see that the ground state energy of the quadratic Hamiltonian in \eqref{eq:quad-Hamil-lwb-general} is $\frac{1}{2}\Tr(E-H)$ where
\[
	E:=(D^{1/2}(D+2K) D^{1/2})^{1/2}, \quad D:=H-K \ge 0.  
\]
Using the formula 
$$x= \frac{2}{\pi} \int_{0}^\infty \frac{x^2}{x^2+t^2} \mathrm{d} t, \quad \forall x\ge 0,$$
and the resolvent identity, we can rewrite 
\begin{align}  \label{eq:deal-I-II-III}
E-D &= \frac{2}{\pi} \int_{0}^\infty  \left(\frac{1}{D^2 + t^2} - \frac{1}{E^2+t^2}\right)t^2 \dd t = \frac{2}{\pi} \int_{0}^\infty \frac{1}{D^2 + t^2} D^{1/2} (2K) D^{1/2} \frac{1}{E^2+t^2} t^2 \dd t  \nn\\
	&=\begin{multlined}[t]
		\frac{2}{\pi} \int_{0}^\infty \frac{1}{D^2 + t^2} D^{1/2} (2K) D^{1/2} \frac{1}{D^2+t^2} t^{2} \dd t\\	
		- \frac{2}{\pi} \int_{0}^\infty \left(\frac{1}{D^2 + t^2} D^{1/2} (2K) D^{1/2}   \right)^2 \frac{1}{D^2+t^2} t^{2} \dd t\\
		+\frac{2}{\pi} \int_{0}^\infty  \left( \frac{1}{D^2 + t^2} D^{1/2} (2K) D^{1/2} \right)^3  \frac{1}{E^2+t^2} t^{2} \dd t
	\end{multlined}\nn\\
	&=: (\mathrm{I}) + (\mathrm{II}) + (\mathrm{III}).
\end{align}

\subsubsection*{Dealing with $(\mathrm{I})$.} Using the cyclicity of the trace and 
\begin{equation*}
	\frac{2}{\pi}\int_{0}^\infty \frac{x t^{2}}{(x^2+t^2)^2} \dd t = \frac{1}{2}, \quad \forall x>0,
\end{equation*}
we have  
\begin{equation} \label{eq:deal-I}
	\Tr{\rm (I)}=     \frac{2}{\pi} \Tr  \int_{0}^\infty \frac{Dt^{2}}{(D^2 + t^2)^2}  (2K)  \dd t = \Tr(K).
\end{equation}
\subsubsection*{Dealing with $(\mathrm{II})$.} Note that $D=H-K>0$ has compact resolvent (since $H$ has compact resolvent and $K$ is compact). Therefore, we can write 
\[ D=\sum_{j} D_j |\varphi_j\rangle \langle \varphi_j| \]
with positive eigenvalues $(D_j)$ and an orthonormal basis of eigenvectors $(\varphi_j)$. Therefore, 
\begin{align*}
	\tr (\mathrm{II}) &=  - \frac{8}{\pi} \Tr  \int_0^\infty \frac{D}{(D^2+t^2)^2} K \frac{D}{D^2+t^2} K t^2 \dd t  \\
	&= - \frac{8}{\pi} \Tr  \int_0^\infty \sum_{i,j} \frac{D_j}{(D_j^2+t^2)^2} |\varphi_j \rangle \langle \varphi_j| K \frac{D_i}{D_i^2+t^2} |\varphi_i\rangle \langle \varphi_i| K t^2 \dd t \\
	&=  - \frac{8}{\pi} \sum_{i,j} |\langle \varphi_i, K \varphi_j\rangle|^2  \int_0^\infty  \frac{D_iD_j }{(D_j^2+t^2)^2(D_i^2+t^2)} t^2 \dd t.  
\end{align*}
Using 
\[
	\frac{8}{\pi} \int_0^\infty \frac{xy }{(x^2+t^2)^2 (y^2+t^2)} t^2 \dd t = \frac{2 y}{(x+y)^2} \le \frac{1}{2x}, \quad \forall x,y>0
\]
we find that
\begin{equation}  \label{eq:deal-II}
	\tr (\mathrm{II}) \ge  - \frac{1}{2} \sum_{ij} |\langle \varphi_i, K \varphi_j\rangle|^2 D_j^{-1} = -\frac{1}{2} \tr \left(K D^{-1} K\right). 
\end{equation}

\subsubsection*{Dealing with $(\mathrm{III})$.} By H\"older's inequality for Schatten norm \cite[Theorem 2.8]{Sim-05},
\begin{align*}
	\Tr{\rm (III)} &= \frac{16}{\pi} \left|  \int_0^\infty \Tr \left( \left( \frac{D}{D^2+t^2} K \right)^3 D^{1/2}  \frac{t^2}{E^2+t^2} D^{-1/2} \right) \dd t \right|\\
	&\le  \frac{16}{\pi} \int_0^\infty   \norm{\frac{D}{D^2+t^2} K}_{\gS^3}^3  \norm{D^{1/2}  \frac{t^2}{E^2+t^2} D^{-1/2}}_{\rm op} \dd t.
\end{align*}
Then by the Araki--Lieb--Thirring inequality \cite{LT-76,Ara-90}, 
\begin{align*}
	\norm{ \frac{D}{D^2+t^2} K}_{\gS^3}^3 &= \Tr \left( \left( K \left( \frac{D}{D^2+t^2} \right)^2 K  \right)^{3/2} \right) = \Tr \left( \left( \frac{D}{D^2+t^2} K^2   \frac{D}{D^2+t^2} \right)^{3/2} \right) \\
	&\le \Tr \left( \left( \frac{D}{D^2+t^2} \right)^{3/2} |K|^3 \left( \frac{D}{D^2+t^2} \right)^{3/2} \right) = \Tr  \left( \frac{D^{3}}{(D^2+t^2)^3} |K|^{3}   \right).
\end{align*}
Here we have used the fact that $|A|=\sqrt{A^*A}$ and $\sqrt{AA^*}$ have the same non-zero eigenvalues (with multiplicity). On the other hand, using $H\ge (1+\eps) \|K\|_{\rm op}$  we find that 
\[ D+2K = H+ K \ge C_\eps^{-1} (H-K)= C_\eps^{-1} D \]
for any large constant $C_\eps$ satisfying $(C_\eps+1)/(C_\eps-1) \le 1+\eps$. Hence, 
\[
E^2 = D^{1/2} (D+2K)D^{1/2} \ge C_\eps^{-1} D^2. 
\]
Reversely, we also have $D+2K \le C_\eps D$, and hence
\[
E^2 = D^{1/2} (D+2K)D^{1/2} \le C_\eps D^2. 
\]
Since the mapping $0\le A\mapsto \sqrt{A}$ is operator monotone, we deduce that $D^{1/2}E^{-1/2}$ and $E^{1/2}D^{-1/2}$ are bounded operators. 
Therefore, 
\begin{align*}
	\norm{D^{1/2} \frac{t^2 }{E^2 +t^2} D^{-1/2}}_{\rm op} &= \norm{D^{1/2} E^{-1/2} \frac{t^2 }{E^2 +t^2} E^{1/2} D^{-1/2}}_{\rm op} \\
	&\le \norm{D^{1/2} E^{-1/2}}_{\rm op}  \norm{\frac{t^2}{E^2+t^2}}_{\rm op} \norm{E^{1/2}D^{-1/2}}_{\rm op} \le C_\eps.
\end{align*}
We conclude that 
\begin{equation}  \label{eq:deal-III}
	|\Tr{\rm (III)}| \le C_\eps  \int_0^\infty \Tr  \left(  \frac{D^3}{(D^2+t^2)^3} |K|^{3} \right) \dd t \le C_\eps \Tr( D^{-2}|K|^3). 
\end{equation}
In the last estimate we have used the identity 
\[
x^{-2}= \frac{16}{3\pi} \int_0^\infty \frac{x^3}{(x^2+t^2)^3} \dd t, \quad \forall x>0. 
\]

\subsubsection*{Conclusion in the trace class case} Inserting \eqref{eq:deal-I}, \eqref{eq:deal-II}, \eqref{eq:deal-III} in \eqref{eq:deal-I-II-III} we find that 
\[
\Tr(E-H) = \Tr(E-D)-\Tr(K) \ge - \frac{1}{2} \Tr(D^{-1}K^2) - C_\eps \Tr( D^{-2}|K|^3).
\]
Let us replace $D=H-K$ by $H$ on the right side. Using $H\ge (1+\eps) \|K\|_{\rm op}$ and the Cauchy--Schwarz inequality we have 
\begin{align*}
	D^2 &= (H-K)^2 = H^2 + K^2 -  HK - KH  \ge (1-\eta) H^2 - (\eta^{-1}-1)K^2 \\
	&\ge (1-\eta) H^2 - (\eta^{-1}-1) (1+\eps)^{-2} H^2 \ge (C_\eps)^{-1} H^2.
\end{align*}
Here the constant $0<\eta<1$ is chosen sufficiently close to $1$ (depending on $\eps$).  Therefore, 
\[
\Tr\left( D^{-2}|K|^3\right) \le C_\eps \Tr\left( H^{-2}|K|^3\right) \le C_\eps \norm{K}_{\rm op}\Tr\left( H^{-2}K^2\right). 
\]
By the resolvent identity and H\"older's inequality for Schatten norm \cite[Theorem 2.8]{Sim-05}, 
\begin{align*}
	\left| \Tr\left(\left(D^{-1}-H^{-1}\right)K^2\right)\right|  &=  \left| \Tr\left(D^{-1} K H^{-1}K^2\right) \right| \\
	 &\le \norm{D^{-1} H}_{\rm op} \norm{H^{-1}K}_{\gS^2}^2  \norm{K}_{\rm op} \le C_\eps \norm{K}_{\rm op}\Tr\left(H^{-2}K^2\right)
\end{align*}
where $\norm{\cdot}_\gS^2$ is the Hilbert--Schmidt norm. Thus  \eqref{eq:quad-Hamil-lwb-general} holds true:
\begin{align*}
	\Tr(E-H)  &\ge - \frac{1}{2} \Tr\left(D^{-1}K^2\right) - C_\eps \Tr\left(D^{-2}|K|^3\right)\\
		&\ge - \frac{1}{2} \Tr\left(H^{-1}K^2\right) - C_\eps \norm{K}_{\rm op}\Tr\left(H^{-2}K^2\right).
\end{align*}

\subsubsection*{Removing the trace class condition} Finally, let us remove the trace class condition on $K$. Recall that $H$ has compact resolvent. For every $n\in \mathbb{N}$ we introduce the spectral projection
$$
P_n = \1(H\le n), \quad Q_n = \1(H>n)
$$
and decompose
$$
K=K_n^{(1)}+K_n^{(2)}, \quad K_{n}^{(1)}= \frac{1}{2}(P_n K + K P_n),\quad K_n^{(2)}= \frac{1}{2}(Q_n K + K Q_n). 
$$

Note that $K_n^{(1)}$ is a self-adjoint finite-rank operator because $P_n$ is finite-rank. Moreover, $K_n^{(1)}$ has a real-valued symmetric kernel because $K$ has the same property and $H$ has real-valued eigenvalues. By the triangle inequality, 
\begin{align*}
\|K_n^{(1)}\|_{\rm op} &\le \frac{1}{2} ( \|P_n K\|_{\rm op} + \|K P_n\|_{\rm op} ) \le \|K\|_{\rm op}, \\
\| H^{-s} K_{n}^{(1)}\|_{\gS^2} &\le \frac{1}{2}( \| H^{-s} P_n K\|_{\gS^2} + \|H^{-s} K P_n\|_{\gS^2}) \le  \| H^{-s} K\|_{\gS^2}, \quad \forall s\ge 1/2.
\end{align*}
Take $\eta\in (0,\eps/2)$. Using $H\ge (1+\eps)\|K\|_{\rm op}$ we find that
\[
\left(1-\eta \right) H \ge \left(1-\frac{\eps}{2}\right) (1+\eps) \norm{K}_{\rm op} \ge  \left(1 + \frac{\eps(1-\eps)}{2}\right) \| K_n^{(1)}\|_{\rm op}.
\]
Applying \eqref{eq:quad-Hamil-lwb-general}  in the trace class case with $(H,K)$ replaced by $((1-\eta) H, K_n^{(1)})$ we get
\begin{align} \label{eq:1-eta-H-K}
	\left(1-\eta \right) &\dGamma(H) + \frac{1}{2} \iint K_1^{(n)}(x,y) (a_x^* a_y^* + a_x a_y) \dd x \dd y \nn\\
		&\ge - \frac{1}{4(1-\eta)}  \| H^{-1/2} K_{n}^{(1)}\|_{\gS^2}^2   - C_\eps \norm{K_1^{(n)}}_{\rm op} \| H^{-1} K_{n}^{(1)}\|_{\gS^2}^2 \nn\\
		&\ge  - \frac{1}{4(1-\eta)} \| H^{-1/2} K\|_{\gS^2}^2   - C_\eps \norm{K}_{\rm op} \| H^{-1} K\|_{\gS^2}^2, \quad \forall n\ge 1.
\end{align}

Next, note that  $K_n^{(2)}$ is also self-adjoint and has a real-valued symmetric kernel. Moreover, $Q_n \to 0$ strongly as $n\to \infty$  (namely $\|Q_n u\|\to 0$ for all $u\in \gK$) since $H$ has compact resolvent. Since $H^{-1/2}K$ is Hilbert-Schmidt, we deduce that 
$$
\| H^{-1/2} K_{n}^{(2)}\|_{\gS^2} \le \frac{1}{2}( \| H^{-1/2} Q_n K\|_{\gS^2} + \|H^{-1/2} K Q_n\|_{\gS^2})  \to 0 \quad \text{ as }n\to \infty. 
$$
Recall that from \cite[Theorem 2]{NNS-16}, the bound \eqref{eq:Bog-half-lower-bound} holds true if $\| H^{-1/2} K H^{-1/2}\|_{\rm op}<1$ and $\| H^{-1/2} K\|_{\gS^2}<\infty$. Using \eqref{eq:Bog-half-lower-bound} with $(H,K)$ replaced by $(\eta H, K_n^{(2)})$, we find that 
\begin{equation} \label{eq:eta-H-K}
	\eta \dGamma(H) + \frac{1}{2} \iint K_n^{(2)}(x,y) (a_x^* a_y^* + a_x a_y) \dd x \dd y \ge - \frac{1}{2\eta} \| H^{-1/2} K_n^{(2)}\|_{\gS^2} \to 0	\end{equation}
	as  $n\to \infty$. Putting \eqref{eq:1-eta-H-K} and \eqref{eq:eta-H-K} together, we find that
$$
	\dGamma(H) + \frac{1}{2} \iint K(x,y) (a_x^* a_y^* + a_x a_y) \dd x \dd y  \ge  - \frac{1}{4(1-\eta)} \| H^{-1/2} K\|_{\gS^2}^2   - C_\eps \norm{K}_{\rm op} \| H^{-1} K\|_{\gS^2}^2.
$$
Taking $\eta\to 0$ we obtain \eqref{eq:quad-Hamil-lwb-general}. This completes the proof of Theorem~\ref{thm:general-bound-HBog}. 
\end{proof}

\subsection{Explicit lower bound for \texorpdfstring{$\mathbb{H}_{\rm Bog}$}{HBog}}

Now we apply Theorem \ref{thm:general-bound-HBog} to compute an explicit lower bound for the quadratic Hamiltonian $\mathbb{H}_{\rm Bog}$ in Lemma \ref{lem:red-qua}.

\begin{lemma} [Lower bound for $\mathbb{H}_{\rm Bog}$] 
	\label{lem:missing_term}
For  $\mathbb{H}_{\rm Bog}$ in Lemma \ref{lem:red-qua} we have
\begin{align} \label{eq:HBog>=fNVN}
 \inf {\rm Spec}(\mathbb{H}_{\rm Bog}) \ge -\frac{N^2}{2}  \int_{\mathbb{R}^{3}} (V_Nf_N(1-f_N) \ast \varphi_{\rm GP}^2)  \varphi_{\rm GP}^2  -C. 
\end{align}
\end{lemma}

\begin{proof} We will write $\varphi=\varphi_{\rm GP}$ for short. Recall the notations $H,K$ in \eqref{eq:H-K-intro}.

\subsubsection*{Lower bound by Theorem \ref{thm:general-bound-HBog}}
Since $\varphi\ge 0$, the kernel $K(x,y)$ is symmetric and real-valued. Thus the operator $K$ is symmetric. It is bounded with 
$\|K\|_{\rm op}\le 8 \pi \ao \norm{\varphi}_{L^\infty}^2$ because for all $g_1,g_2\in \gH_+$ we have
\begin{align}
	\left|\langle g_1,K g_2 \rangle \right| &= \left| \iint \overline{g_1(x)} \varphi(x) (Nf_NV_N(x-y)) \varphi(y) g_2(y) \dd x \dd y \right| \nn \\
	&\le \norm{\varphi}_{L^\infty}^2 \norm{g_1}_{L^2} \norm{g_2}_{L^2} \norm{Nf_NV_N}_{L^1}= 8 \pi \ao \norm{\varphi}_{L^\infty}^2 \norm{g_1}_{L^2} \norm{g_2}_{L^2}. \label{eq:K-op}
\end{align}

On the other hand, since $\mu_2>\mu$ we have  
\[ H=Q(-\Delta+V_{\rm ext}-\mu)Q \ge \mu_2-\mu + 8\pi \ao \norm{\varphi}_{L^\infty}^2  \ge (1+\eps) \norm{K}_{\rm op}\quad \text{ on } \gH_+ \]
for a small constant $\eps>0$ independent of $N$. Moreover, $H$ has compact resolvent since $V_{\rm ext}(x)\to \infty$ as $|x|\to \infty$. Thus we can apply Theorem \ref{thm:general-bound-HBog} and obtain 
\begin{equation} \label{eq:lwb-appl}
	\inf {\rm Spec}(\mathbb{H}_{\rm Bog}) \ge -\frac{1}{4} \Tr_{\gH_+}\left(H^{-1}K^2\right) - C \Tr_{\gH_+}\left(H^{-2}K^2\right). 
\end{equation}

\subsubsection*{Replacing $H$ and $K$ by $1-\Delta$ and $\widetilde K$} We can interpret $ H = Q(-\Delta+V_{\rm ext}-\mu) Q$ as an operator on $L^2(\R^3)$. Then using $V_{\rm ext} \ge 0$ and $Q=1-|\varphi\rangle \langle \varphi|$ we have 
\begin{align} \label{eq:comp-1a}
	H \ge Q(-\Delta) Q -\mu &=-\Delta + |\varphi\rangle \langle \Delta \varphi| + |\Delta \varphi\rangle \langle  \varphi| + \norm{\nabla \varphi}_{L^2(\R^3)}^2 |\varphi\rangle \langle  \varphi| -\mu \nn\\
	&\ge -\Delta - 2\norm{\Delta \varphi}_{L^2} -\mu  \quad \text{ on } L^2(\R^3).
\end{align}
Moreover, by  \eqref{eq:H^2>Delta^2} and the Cauchy-Schwarz inequality, 
\begin{align} \label{eq:comp-2a}
	H^2 &\ge Q(-\Delta+V_{\rm ext}-\mu)^2 Q - \norm{Q(-\Delta+V_{\rm ext}-\mu) \varphi}_{L^2(\R^3)}^2 \nn\\
	&\ge \frac{1}{2} Q(\Delta)^2 Q - C = \frac{1}{2}(\Delta)^2 - \frac{1}{2} \big(|\varphi\rangle \langle \Delta \varphi| \Delta +  \Delta |\Delta\varphi\rangle \langle  \varphi| \big)  +\frac{1}{2} \norm{\Delta \varphi}^2_{L^2} - C \nn\\
	&\ge \frac{1}{4}(\Delta)^2 - C \quad \text{ on } L^2(\R^3).
\end{align}
Since $H_{|\gH_+}$ is strictly positive, we also have, for any large constant  $C_0>0$, 
\begin{equation} \label{eq:comp-2b}
	C_0^2 H^2 \ge H^2+C_0 \quad \text{ on }\gH_+.
\end{equation}

Recall that the mapping $t\mapsto t^{-1}$ is operator monotone for $t>0$. Moreover, if $A$ is a self-adjoint positive operator on $L^2(\R^3)$ that commutes with $Q$, then 
\begin{equation} \label{QAQ-A}
	Q (QAQ)_{|\gH_+}^{-1} Q = Q A^{-1}Q\quad \text{ on } L^2(\R^3)
\end{equation}
by Spectral Theorem. Therefore, 
\begin{align} \label{eq:comp-2c}
	Q H_{|\gH_+}^{-2} Q &\le C_0^2 Q (H^2+C_0)_{|\gH_+}^{-1} Q  \nn\\
	&\le C Q (H^2+C_0)^{-1} Q \le C Q(1-\Delta)^{-2} Q  \quad \text{ on } L^2(\R^3).
\end{align}
Similarly, 
\begin{align} \label{eq:comp-1}
	QH_{|\gH_+}^{-1}Q &= Q (H+C)_{|\gH_+}^{-1}Q + C Q (H(H+C))_{|\gH_+}^{-1} Q \nn\\
	&\le Q ( H+C)^{-1} Q + C Q H_{|\gH_+}^{-2} Q \nn\\
	&\le Q (1-\Delta)^{-1}Q + C Q(1-\Delta)^{-2} Q \quad \text{ on } L^2(\R^3).
\end{align}

Next we replace $K$ by $\widetilde K$. Following \eqref{eq:K-op} we have $\|\widetilde K\|_{\rm op} \le C$ where $K$ is the operator on $L^2(\R^3)$ with kernel 
$\widetilde K(x,y)$. Using $K^2\le Q\widetilde K^2 Q$ on $\gH_+$ and \eqref{eq:comp-2c} we can estimate
\begin{align} 
	\Tr_{\gH_+}\left(H^{-2}K^2\right) &\le \Tr_{\gH_+}\left(H^{-2}Q\widetilde K^2Q\right)= \Tr_{L^2(\R^3)}\left( QH^{-2}_{|\gH_+}Q \widetilde K^2\right) \nn\\ 
	&\le \begin{multlined}[t]
		C \Tr\left(Q (1-\Delta)^{-2}Q \widetilde K^2\right) = C \Tr\left( (1-\Delta)^{-2}Q \widetilde K^2 Q\right)\\
		= C\Tr\left( (1-\Delta)^{-2} \left( \widetilde K^2 - |\varphi \rangle \langle \widetilde K^2 \varphi| - |\widetilde K^2\varphi \rangle \langle   \varphi| + \|\widetilde K \varphi\|_{L^2(\R^3)}^2 \right)\right)
	\end{multlined}\nn\\
	&\le C \Tr\left( (1-\Delta)^{-2} \widetilde K^2 \right) + C. \label{eq:comp-2}
\end{align}
Similarly, from \eqref{eq:comp-1a} we deduce that
\begin{align} 
	\Tr_{\gH_+}\left(H^{-1}K^2\right) &\le \Tr_{\gH_+}\left(QH^{-1}Q\widetilde K^2\right) \le \Tr\left( \left( Q (1-\Delta)^{-1}Q + C Q(1-\Delta)^{-2} Q \right) \widetilde K^2\right)\nn\\
	&\le \Tr\left((1-\Delta)^{-1} \widetilde K^2\right) + C \Tr\left( (1-\Delta)^{-2} \widetilde K^2 \right) + C. \label{eq:comp-1}
\end{align}
Thus \eqref{eq:lwb-appl} reduces to 
\begin{equation} \label{eq:lwb-appl-aaa}
	\inf {\rm Spec}(\mathbb{H}_{\rm Bog}) \ge -\frac{1}{4}\Tr\left((1-\Delta)^{-1} \widetilde K^2\right) -C \Tr\left( (1-\Delta)^{-2} \widetilde K^2 \right) - C. 
\end{equation}

\subsubsection*{Evaluation of traces in \eqref{eq:lwb-appl-aaa}} Note that the operator $\widetilde{K}$ with kernel $\varphi(x) N{f_NV_N}(x-y) \varphi(y)$ can be written as   
\begin{equation} \label{eq:K-phi-g}
	\widetilde{K}= \varphi(x) N\widehat{f_NV_N}(p) \varphi(x) \quad \text{ on }L^2(\R^3)
\end{equation}
where $\varphi(x)$ and $v(p)$ are the multiplication operators on the position and momentum spaces (the derivation of \eqref{eq:K-phi-g} uses $\widehat{u*v}=\widehat u \widehat v$). Recall the Kato--Seiler--Simon inequality on Schatten norms   \cite[Theorem 4.1]{Sim-05}:
\begin{align}\label{eq:KSS}
	\norm{u(x) v(p)}_{\gS^r} \le C_{d,r}\norm{u}_{L^r(\R^d)}\norm{v}_{L^r(\R^d)}, \quad 2\le r <\infty. 
\end{align}
Consequently, 
\begin{align} \label{eq:lwb-appl-1}
	\tr\Big( (1-\Delta)^{-2} &\widetilde K^2\Big)		= \tr \left( \varphi(x) N\widehat{(f_NV_N)}(p) \varphi(x) \frac{1}{(1+p^2)^2}  \varphi(x) N\widehat{(f_NV_N)} (p) \varphi(x) \right) \nonumber \\
		&\le  \norm{\varphi}_{L^\infty(\mathbb{R}^{3})}^2 \|N\widehat{f_N V_N}\|_{L^\infty(\mathbb{R}^{3})}^2  \norm{\varphi(x) (p^2 + 1)^{-1}}^2_{\gS^2} \nonumber \\
		&\le C \norm{\varphi}_{L^\infty(\mathbb{R}^{3})}^2 \|N\widehat{f_N V_N}\|_{L^\infty(\mathbb{R}^{3})}^2 \norm{\varphi}_{L^2(\R^3)}^2 \norm{(p^2 + 1)^{-1}}_{L^2(\R^3)}^2\le C.
\end{align}
Here we have used H\"older's inequality for Schatten spaces and $\|N\widehat{f_N V_N}\|_{L^\infty} \le  8\pi \ao$.

Next, consider  
\[
\Tr\left((1-\Delta)^{-1}\widetilde K^2\right) = \tr\left(\varphi(x) N\widehat{(f_NV_N)} (p) \varphi(x) \frac{1}{1+p^2} \varphi(x) N\widehat{(f_NV_N)} (p) \varphi(x) \right). 
\]
Let us decompose 
\begin{multline*}
	2 \varphi(x) (1+p^2)^{-1} \varphi(x) -  \varphi^2(x)  (1+p^2)^{-1}  - (1+p^2)^{-1}  \varphi^2(x) \\
	\begin{aligned}[t]
		&= - \left[ \varphi(x), \left[ \varphi(x), (1+p^2)^{-1}\right]\right] = \left[ \varphi(x), (1+p^2)^{-1} [\varphi(x),p^2] (1+p^2)^{-1}\right] \\
		&=\begin{multlined}[t]
			\left[ \varphi(x), (1+p^2)^{-1}\right] [\varphi(x),p^2] (1+p^2)^{-1} + (1+p^2)^{-1} \left[ \varphi(x),  [\varphi(x),p^2]\right] (1+p^2)^{-1} \\
			\qquad+ (1+p^2)^{-1} [\varphi(x),p^2] \left[ \varphi(x),  (1+p^2)^{-1}\right]
		\end{multlined}\\
		&=\begin{multlined}[t]
			-2 (1+p^2)^{-1} [\varphi(x),p^2] (1+p^2)^{-1} [\varphi(x),p^2] (1+p^2)^{-1} \\
			- 2 (1+p^2)^{-1} |\nabla \varphi(x)|^2 (1+p^2)^{-1},
		\end{multlined}
	\end{aligned}
\end{multline*}
where we have used 
\[ [\varphi(x),(p^2+1)^{-1}] = -  (p^2+1)^{-1} [\varphi(x),p^2](p^2+1)^{-1} \]
and the IMS formula \eqref{eq:IMS-formula} for $[\varphi(x), [\varphi(x),p^2]]$. This gives 
\begin{multline}\label{eq:1-2-3-second}
	\varphi(x) N\widehat{f_NV_N} (p)  \varphi(x) \frac{1}{1+p^2} \varphi(x) N\widehat{f_NV_N}(p) \varphi(x) \\
	\begin{aligned}[b]
		&= \begin{multlined}[t]
			\frac{1}{2} \varphi(x) N\widehat{f_NV_N} (p)  \left( \varphi^2(x)  \frac{1}{1+p^2}  +   \frac{1}{1+p^2}  \varphi^2(x) \right)  N\widehat{f_NV_N}(p) \varphi(x)\\
				- \varphi(x) N\widehat{f_NV_N} (p) \frac{1}{1+p^2} [\varphi(x),p^2] \frac{1}{1+p^2} [\varphi(x),p^2] \frac{1}{1+p^2} N\widehat{f_NV_N}(p) \varphi(x)\\
				- \varphi(x) N\widehat{f_NV_N} (p) \frac{1}{1+p^2} |\nabla \varphi(x)|^2 \frac{1}{1+p^2} N\widehat{f_NV_N}(p) \varphi(x)
		\end{multlined}\\
		&=: (\mathrm{I}) + (\mathrm{II}) + (\mathrm{III}).
	\end{aligned}
\end{multline}

\subsubsection*{Dealing with $(\mathrm{I})$.}  For the main term (I), we write  
\begin{align*}
\tr (\mathrm{I}) 
	&= \Re  \Tr \left( \varphi^2(x) N\widehat{f_NV_N} (p) \varphi^2(x)  \frac{N\widehat{f_NV_N}(p)}{1+p^2}   \right)\\
	&= \begin{multlined}[t]
		\Re  \Tr \left( \varphi^2(x) N\widehat{f_NV_N} (p) \varphi^2(x)  \frac{N\widehat{f_NV_N}(p)}{p^2}  \right) \\
		- \Re  \Tr \left(  N\widehat{f_NV_N} (p) \varphi^2(x)  \frac{N\widehat{f_NV_N}(p) }{p^2(1+p^2)}  \varphi^2(x) \right). 
	\end{multlined}
\end{align*}
The first term can be computed exactly using the scattering equation \eqref{eq:scatering-eq-N-omega}
\begin{multline*}
\Re  \Tr \left( \varphi^2(x) N\widehat{f_NV_N} (p) \varphi^2(x)  \frac{1}{p^2}  N\widehat{f_NV_N}(p) \right)\\
= 2 \Re \Tr \left( \varphi^2(x) N\widehat{f_NV_N} (p) \varphi^2(x)  N\widehat {(1-f_N)}(p) \right) = 2  N^2  \int_{\mathbb{R}^{3}} (V_Nf_N(1-f_N) \ast \varphi^2)  \varphi^2. 
\end{multline*}
Here we have used the following identity 
\begin{align}\label{eq:tr-int}
	\Tr|\left(\overline{\varphi_1(x)} \overline{ \widehat{g}_1(p)} \varphi_2(x) \widehat g_2(p) \right) &= \left\langle \widehat{g}_1(p) \varphi_1(x), \varphi_2(x) \widehat g_2(p) \right\rangle_{\gS^2}\nn\\
	&= \left\langle (\widehat{g}_1(p) \varphi_1(x)) (\cdot,\cdot), (\varphi_2(x) \widehat g_2(p)) (\cdot, \cdot) \right\rangle_{L^2(\R^3\times \R^3)} \nn\\
	&= \iint \overline{g_1(y-z) \varphi_1(z)} \varphi_2(y) g_2(y-z) \dd y \dd z.
\end{align}
This is based on the equality between the Hilbert--Schmidt norm of operators and the $L^2$-norm of operator kernels (the kernel of operator $\widehat{g}_1(p) \varphi_1(x)$ is $g_1(y-z)\varphi(z)$, similarly to \eqref{eq:K-phi-g}). Here in our case $\varphi(x)\ge 0$ and $\widehat{f_NV_N}(p)$ is real-valued since $V_Nf_N$ is radial.

The second term can be estimated by H\"older's and the  Kato--Seiler--Simon inequalities:
\begin{align*}
	&\left| \Tr \left(  N\widehat{f_NV_N} (p) \varphi^2(x)  \frac{N\widehat{f_NV_N}(p)}{p^2(1+p^2)}   \varphi^2(x) \right) \right| \le \norm{N\widehat{f_NV_N}}_{L^\infty}^2 \norm{\varphi^2(x) \sqrt{\frac{1}{p^2(1+p^2)}}}_{\gS^2}^2 \\
	 &\qquad\qquad\qquad\qquad\qquad \qquad\qquad \quad  \le C\norm{N\widehat{f_NV_N}}_{L^\infty}^2  \norm{\varphi^2}_{L^2}^2 \norm{\sqrt{\frac{1}{p^2(1+p^2)}}}_{L^2}^2 \\
	&\qquad\qquad\qquad\qquad\qquad \qquad\qquad \quad \le C\norm{N\widehat{f_NV_N}}_{L^\infty}^2 \norm{\varphi}_{L^4}^4 \norm{\frac{1}{p^2(1+p^2)}}_{L^1} \le C.   
\end{align*}
Here we used again  $\| N \widehat{f_NV_N} \|_{L^\infty}  \le 8\pi \ao$. Thus
\begin{align*}
	\tr (\mathrm{I}) = 2 N^2  \int_{\mathbb{R}^{3}} (V_Nf_N(1-f_N) \ast \varphi^2)  \varphi^2 + O(1). 
\end{align*}

\subsubsection*{Dealing with $(\mathrm{II})$.}  By expanding further
\[ [\varphi(x),p^2] =  (\Delta \varphi)(x) +   2 (\nabla \varphi(x)) \cdot \nabla \]
and using the triangle inequality we have 
\begin{align*}
	& |\Tr(\mathrm{II})| = \left| \Tr \left( \varphi(x) \frac{N\widehat{f_NV_N} (p) }{1+p^2} [\varphi(x),p^2] \frac{1}{1+p^2} [\varphi(x),p^2] \frac{1}{1+p^2} N\widehat{f_NV_N}(p) \varphi(x) \right) \right| \\
	&\begin{multlined}[t]
		\le\left| \Tr \left( \varphi(x) \frac{N\widehat{f_NV_N} (p) }{1+p^2} (\Delta \varphi(x)) \frac{1}{1+p^2} (\Delta \varphi(x)) \frac{1}{1+p^2} N\widehat{f_NV_N}(p) \varphi(x) \right) \right| \\
		+ 4 \left| \Tr \left( \varphi(x) \frac{N\widehat{f_NV_N} (p) }{1+p^2} (\Delta \varphi(x)) \frac{1}{1+p^2} ((\nabla \varphi(x)) \cdot \nabla) \frac{1}{1+p^2} N\widehat{f_NV_N}(p) \varphi(x) \right) \right|\\
		+ 4\left| \Tr \left( \varphi(x) \frac{N\widehat{f_NV_N} (p) }{1+p^2} ((\nabla \varphi(x)) \cdot \nabla) \frac{1}{1+p^2} ((\nabla \varphi(x)) \cdot \nabla) \frac{1}{1+p^2} N\widehat{f_NV_N}(p) \varphi(x) \right) \right|.
	\end{multlined}
\end{align*}
Then by H\"older's and the Kato--Seiler--Simon inequalities,
\begingroup
\allowdisplaybreaks
\begin{align*}
 	|\Tr(\mathrm{II})| & \le\norm{\varphi}_{L^\infty}^2 \norm{N \widehat{f_NV_N}}_{L^\infty}^2 \norm{(\Delta \varphi(x)) \frac{1}{1+p^2}}_{\rm \gS^2}^2\\
		&\quad + 4 \norm{\varphi}_{L^\infty}^2  \norm{N \widehat{f_NV_N}}_{L^\infty}^2 \norm{\frac{1}{1+p^2} (\Delta \varphi(x))}_{\rm \gS^2} \norm{\frac{1}{1+p^2}  |\nabla \varphi(x)|}_{\rm \gS^2} \norm{ |\nabla| \frac{1}{1+p^2}}_{\rm op} \\
			&\qquad +4  \norm{ \varphi(x) \frac{1}{\sqrt{1+p^2}}}_{\gS^4} \norm{N \widehat{f_NV_N}}_{L^\infty} \norm{  \frac{1}{\sqrt{1+p^2}} |\nabla \varphi(x)|}_{\gS^4} \norm{ |\nabla| \frac{1}{\sqrt{1+p^2}} }_{\rm op}  \times \\
			& \qquad\quad \times \norm{  \frac{1}{\sqrt{1+p^2}} |\nabla \varphi(x)|}_{\gS^4}  \norm{ |\nabla| \frac{1}{\sqrt{1+p^2}} }_{\rm op} \norm{N \widehat{f_NV_N}}_{L^\infty} \norm{ \frac{1}{\sqrt{1+p^2}} \varphi(x) }_{\gS^4}\\
		&\le C\norm{\varphi}_{L^\infty}^2 \norm{N \widehat{f_NV_N}}_{L^\infty}^2 \norm{\Delta \varphi}_{L^2}^2 \norm{ \frac{1}{1+p^2}}_{L^2}^2\\
		&\quad + 4 \norm{\varphi}_{L^\infty}^2  \norm{N \widehat{f_NV_N}}_{L^\infty}^2 \norm{  \frac{1}{1+p^2}}_{L^2}^2 \norm{\Delta \varphi}_{L^2} \norm{\nabla \varphi(x)}_{L^2} \norm{ \frac{|p|}{1+p^2}}_{\rm op} \\
		&\qquad + 4 \norm{N \widehat{f_NV_N}}_{L^\infty}^2 \norm{\varphi}_{L^4}^2 \norm{\nabla \varphi}_{L^4}^2 \norm{  \frac{1}{\sqrt{1+p^2}}}_{L^4}^4 \norm{ \frac{|p|}{1+p^2}}_{\rm op}^2 \le C.
\end{align*}
\endgroup

\subsubsection*{Dealing with $(\mathrm{III})$} This term is negative and can be ignored for an upper bound.  Nevertheless, we can bound it by H\"older's and the  Kato--Seiler--Simon inequalities, 
\begin{align*}
	|\tr({\rm III})| &= \left| \Tr \left( \varphi(x) N\widehat{f_NV_N} (p) \frac{1}{1+p^2} |\nabla \varphi(x)|^2 \frac{1}{1+p^2} N\widehat{f_NV_N}(p) \varphi(x) \right) \right| \\
	&\leq \norm{\varphi}_{L^\infty(\mathbb{R}^{3})}^2 \norm{N\widehat{f_NV_N}}_{L^\infty(\mathbb{R}^{3})}^2 \norm{|\nabla \varphi(x)| \frac{1}{p^2+1}}_{\rm \gS^2}^2 \\
	& \le \norm{\varphi}_{L^\infty(\mathbb{R}^{3})}^2 \norm{N\widehat{f_NV_N}}_{L^\infty(\mathbb{R}^{3})}^2  \norm{\nabla \varphi}_{L^2(\R^3)}^2 \norm{\frac{1}{p^2+1}}_{L^2(\R^3)}^2 \le C.
\end{align*}

In summary, we deduce from \eqref{eq:1-2-3-second} that 
\begin{align} \label{eq:lwb-appl-2}
	\Tr((1-\Delta)^{-1} \widetilde K^2) &= \Tr \left( \varphi(x) N\widehat{f_NV_N} (p)  \varphi(x) \frac{1}{1+p^2} \varphi(x) N\widehat{f_NV_N}(p) \varphi(x)  \right) \nn\\
	&= 2  N^2  \int_{\mathbb{R}^{3}} (V_Nf_N(1-f_N) \ast \varphi^2)  \varphi^2 + O(1). 
\end{align}
Inserting \eqref{eq:lwb-appl-1}  and \eqref{eq:lwb-appl-2} in \eqref{eq:lwb-appl-aaa} we obtain the desired estimate \eqref{eq:HBog>=fNVN}: 
\[
 \inf {\rm Spec}(\mathbb{H}_{\rm Bog}) \ge -\frac{1}{2}N^2  \int_{\mathbb{R}^{3}} (V_Nf_N(1-f_N) \ast \varphi^2)  \varphi^2  -C. 
\]
 \end{proof}

\subsection{Conclusion of lower bound} 

\begin{proof}[Proof of Lemma \ref{lem:lwb}] From Lemma \ref{lem:red-qua} and Lemma \ref{lem:missing_term}, we have
\begin{align*} 
	H_N &\ge \begin{multlined}[t]
		N  \int_{\mathbb{R}^{3}} \left(|\nabla \varphi_{\rm GP} |^2 + \Vext |\varphi_{\rm GP}|^2 \right)  + \frac{N^2}{2} \int_{\R^3} \left ((((2f_N-f_N^2)V_N)*\varphi_{\rm GP}^2)\varphi_{\rm GP}^2 \right)\\
		+  (\mu- \mu_1) \cN_+ +  \inf {\rm Spec} (\mathbb{H}_{\rm Bog})- C
	\end{multlined}\\
	&\ge \begin{multlined}[t]
		N\int_{\mathbb{R}^{3}} \left(|\nabla \varphi_{\rm GP} |^2 + \Vext |\varphi_{\rm GP}|^2 \right)  + \frac{N^2}{2} \int_{\R^3} \left ((((2f_N-f_N^2)V_N)*\varphi_{\rm GP}^2)\varphi_{\rm GP}^2 \right)\\
		+ (\mu- \mu_1) \cN_+ -\frac{1}{2}N^2  \int_{\mathbb{R}^{3}} (V_Nf_N(1-f_N) \ast \varphi^2)  \varphi^2  -C
	\end{multlined}\\
	& = \begin{multlined}[t]
		N \int_{\mathbb{R}^{3}} \left(|\nabla \varphi_{\rm GP} |^2 + \Vext |\varphi_{\rm GP}|^2 \right)  + \frac{N^2}{2} \int_{\R^3} \left (((f_N V_N)*\varphi_{\rm GP}^2)\varphi_{\rm GP}^2 \right) \\
		+ (\mu- \mu_1) \cN_+  -C.
	\end{multlined}
\end{align*}
It remains to show that
\begin{equation} \label{eq:fnvn-delta}
	\frac{N^2}{2} \int_{\R^3} \left (((f_N V_N)*\varphi_{\rm GP}^2)\varphi_{\rm GP}^2 \right) = N 4\pi \ao \int_{\R^3} |\varphi_{\rm GP}|^4 + O(1). 
\end{equation}
In fact, we have 
\begin{multline*}
	\left| N \int \left (((f_N V_N)*\varphi_{\rm GP}^2)\varphi_{\rm GP}^2 \right) - 8\pi \ao \int |\varphi_{\rm GP}|^4\right|\\
	\begin{aligned}[b]
		&= \left| \int N \widehat {f_NV_N}(k) |\widehat {\varphi_{\rm GP}^2}(k)|^2 \dd k - \widehat {fV}(0) \int |\widehat {\varphi_{\rm GP}^2}(k)|^2 \dd k \right|\\
		&=  \left| \int  \left( \widehat {fV}(k/N) - \widehat {fV}(0) \right) |\widehat {\varphi_{\rm GP}^2}(k)|^2 \dd k \right|\\
		&\le \norm{\nabla_k \widehat {fV}}_{L^\infty} \int  |k/N| |\widehat {\varphi_{\rm GP}^2}(k)|^2 \dd k\le C N^{-1} \norm{|x|fV}_{L^1} \norm{\varphi_{\rm GP}^2}_{H^{1/2}} \le C N^{-1}.
	\end{aligned}
\end{multline*}
Thus \eqref{eq:fnvn-delta} holds true. Hence we find that 
\begin{multline*} 
	H_N \ge N \int_{\mathbb{R}^{3}} \left(|\nabla \varphi_{\rm GP} |^2 + \Vext |\varphi_{\rm GP}|^2 \right) + 4\pi \ao N \int_{\R^3} |\varphi_{\rm GP}|^4\\
	+ (\mu- \mu_1) \cN_+  -C= N e_{\rm GP} + (\mu- \mu_1) \cN_+  -C.
\end{multline*}
Since $\mu>\mu_1$ and $\cN_+=\sum_{i=1}^N Q_{x_i}$, the proof of Lemma \ref{lem:lwb} is complete. 
\end{proof}

\section{Upper bound} \label{sec:upp}

In this section prove the missing energy upper bound. 

\subsection{Construction of the trial state}
Let us explain the construction of the trial state. In the Fock space setting, it is known \cite[Appendix A]{BPS-16} that we can reach the energy $Ne_{\rm GP} + O(1)$ using trial states of the form 
\[  W(\sqrt{N}\varphi_{\rm GP}) \Gamma' W(\sqrt{N}\varphi_{\rm GP})^* \]
where $W(g)=e^{a(g)-a^*(g)}$ is the Weyl operator and $\Gamma'$ is an appropriate quasi-free state. In the following, we will adapt this construction to the $N$-particle Hilbert space. We will use the unitary operator $U_N$ introduced  in  \cite{LNSS-15} instead of the Weyl operator and modify the quasi-free state slightly. Denote 
\[ Q=\1-|\varphi_{\rm GP} \rangle\langle \varphi_{\rm GP}|, \quad \mathcal{H}_+ = QL^2(\R^3). \]
As explained in \cite{LNSS-15}, any function $\Psi_N\in L^2(\R^3)^{\otimes_s N}$ admits a unique decomposition
\begin{equation*}
	\Psi_N = \varphi^{\otimes N} \xi_0 + \varphi^{\otimes N-1} \otimes_s \xi_1 + \varphi^{\otimes N-2} \otimes_s \xi_2 + ... + \xi_N
\end{equation*}
with $\xi_k \in \mathcal{H}_+^{\otimes_s k}$ (with the convention that $\xi_0 \in \mathbb{C}$). This defines a unitary map $U_N$ from $L^2(\R^3)^{\otimes_s N}$ to $\cF^{\le N}(\mathcal{H}_+)$, the truncated Fock space with particle number $\cN \le N$, by 
\begin{equation} \label{eq:def-UN}
	U_N \left(  \sum_{k} \varphi^{\otimes N-k} \otimes_s \xi_k \right) = \bigoplus_{k=0}^N \xi_k.
\end{equation}

Next, let $k$ be the Hilbert--Schmidt operator on $L^2(\R^3)$ with kernel   
\[ k(x,y)= \varphi_{\rm GP}(x) N (1-f_N(x-y))  \varphi_{\rm GP}(y), \]
with $f_N$ the scattering solution in \eqref{eq:scatering-eq-N}. Define $\gamma:\gH_+\to \gH_+$, $\alpha:\gH_+\to \gH_+^*\equiv \overline{\gH_+}$ by 
\[ \gamma=Qk^2Q, \quad \text{and} \quad \alpha=\overline{Q} kQ. \]
Then we have $\gamma\ge 0$, $\Tr \gamma\le \Tr k^2 <\infty$, $\alpha^*=Qk\overline{Q}=\overline{\alpha}$ and  for all $g_1,g_2\in \gH_+$
\begin{align*}
	& \left \langle
	\begin{pmatrix}
		g_1 \\
		\overline{g_2} 
	\end{pmatrix},
	\begin{pmatrix}
		\gamma &  \alpha^* \\
		\alpha & 1 + \overline{\gamma} 
	\end{pmatrix}
	\begin{pmatrix}
		g_1 \\
		\overline{g_2}
	\end{pmatrix}
	\right\rangle_{\gH_+\oplus \gH_+^*} = \left\langle g_1, k^2 g_1 \right\rangle + \left\langle g_2, \left(1+k^2\right) g_2\right\rangle + 2\Re  \left\langle  \overline{g_2}, k {g_1} \right\rangle \ge 0
\end{align*}
by the Cauchy--Schwarz inequality. Thus $(\gamma,\alpha)$ satisfies \eqref{eq:1-pdm-quasi}. Hence, there exists a unique (mixed) quasi-free state $\Gamma$ on the excited Fock space $\cF(\gH_+)$ such that $(\gamma,\alpha)$ are its one-body density matrices, namely 
\begin{equation} \label{eq:def-G+}
	\left\langle g_1, k^2 g_2\right\rangle = \left\langle a^*(g_2)a(g_1) \right\rangle_\Gamma, \quad  \left\langle \overline{g_1}, k g_2\right\rangle = \left\langle a^*(g_2)a^*(g_1) \right\rangle_\Gamma, \quad \forall g_1,g_2\in \gH_+.
\end{equation}


In this section we will prove
\begin{lemma}[Upper bound] \label{lem:upp} Let $V_{\rm ext}$, $V$ as in Theorem~\ref{thm:main} (but without the technical condition \eqref{eq:gap-condition}). Let $U_N=U_N(\varphi_{\rm GP})$ be as in \eqref{eq:def-UN} and let $\Gamma$ be as in  \eqref{eq:def-G+}. Let $\1^{\le N}=\1(\cN\le N)$ be 
the truncation on the particle number operator. Then  
\[
	\Gamma_N := U_N^* \1^{\le N} \Gamma \1^{\le N} U_N
\]
is a non-negative operator on $L^2(\R^3)^{\otimes_s N}$ with $|1- \Tr \Gamma_N| \le C_s N^{-s}$ for any $s \ge 1$, and 
\[ \Tr(H_N \Gamma_N)  \le N e_{\rm GP} + C \]
with the Hamiltonian $H_N$ in \eqref{eq:HN}. Consequently, 
\[ E_N\le \frac{\Tr(H_N \Gamma_N)}{\Tr \Gamma_N} \le Ne_{\rm GP} + C. \]
\end{lemma}

\begin{remark} In the Fock space setting in \cite[Appendix A]{BPS-16}, the quasi-free state $\Gamma'$ is constructed using an explicit Bogoliubov transformation of the form 
\begin{equation*} 
T_0=\exp\left(\frac{1}{2} \iint_{\mathbb{R}^{3}\times\mathbb{R}^{3}} k(x,y) (a_x^* a_y^* - a_x a_y) \dd x \dd y \right).
\end{equation*}
Its action on creation and annihilation operators is given for any $g\in L^2(\mathbb{R}^{3})$ by
\begin{equation*}\label{eq:action_of_T_on_a}
T_0^* a^*(g) T_0 = a^*(\ch(k) g) + a(\sh(k) \overline{g}),
\end{equation*}
where 
\begin{equation*}
\ch(k) = \sum_{n\geq0} \frac{(k \overline{k})^n}{(2n)!} \quad \textrm{ and } \quad \sh(k) = \sum_{n\geq0} \frac{(k \overline{k})^nk}{(2n+1)!}.
\end{equation*}
The one-body density matrices of $\Gamma'$ can be computed in terms of $\ch(k)$ and $\sh(k)$. Our construction of the quasi-free state $\Gamma$ is slightly  different as its one-body density matrices are given exactly in terms of $k$. This makes the energy computation easier. 
\end{remark}

We divide the proof of Lemma \eqref{lem:upp} into several steps.

\subsection{Operator bound on Fock space} First, we analyse the action of the unitary transformation $U_N$. 

\begin{lemma}[Operator bound on Fock space] \label{lem:UHU*} We have the operator inequality 
\begin{align}
	\1^{\le N} U_N H_N U_N^*\1^{\le N} \le \1_{\cF_+} (\mathcal{G}_N + C(\cN+1)^6 )   \1_{\cF_+} \quad \text{on }\cF(\gH_+),
\end{align}
where $\mathcal{G}_N$ is the following operator on the full Fock space $\cF(L^2(\R^3))$:
\begin{equation*}
	\mathcal{G}_N = \begin{multlined}[t]
		N \int \left(|\nabla\varphi_{\rm GP}|^2 + \Vext |\varphi_{\rm GP}|^2 \right) + \frac{N^2}{2} \int \left(V_N \ast \varphi_{\rm GP}^2\right)  \varphi_{\rm GP}^2  \\
		+  \sqrt{N} \left(  a  \left(  \left(-\Delta + V_{\rm ext} + NV_N*\varphi_{\rm GP}^2\right)\varphi_{\rm GP}\right) + {\rm h.c.}\right) \\
		+ \dd\Gamma (-\Delta+V_{\rm ext}) + \frac{N}{2} \iint V_N (x-y) \varphi_{\rm GP}(x) \varphi_{\rm GP}(y) (a^*_x a^*_y + {\rm h.c.})   \dd x\, \dd y  \\
		+ \sqrt{N} \iint V_N(x-y) \varphi_{\rm GP}(x) (a^*_y a_{x} a_{y} + {\rm h.c.} ) \dd x \dd y \\
		+ \frac{1+CN^{-1}}{2}  \iint V_N(x-y) a^*_x a^*_y a_{x} a_{y} \dd x \, \dd y. 
	\end{multlined}
\end{equation*}
\end{lemma}

\begin{proof} Let us write $\varphi=\varphi_{\rm GP}$ for short.  After a straightforward computation as in \cite[Section 4]{LNSS-15} (see also \cite[Appendix B]{LNS-15}) using \eqref{eq:2nd-Q} and the rules
\begin{equation*}
	U_N a^*(\varphi)a(g_1) U_N^* =  \sqrt{N-\cN_+} a(g_1),  \quad U_N a^*(g_1)a(g_2) U_N^* = a^*(g_1)a(g_2), \quad \forall g_1,g_2\in \gH_+
\end{equation*}
we obtain 
\begin{equation} \label{eq:UHU-wG}
	\1^{\le N} U_N H_N U_N^* \1^{\le N} = \1_{\cF_+}  \1^{\le N} \widetilde{\mathcal{G}}_N  \1^{\le N}  \1_{\cF_+}  \quad \text{ on }  \cF_+,
\end{equation}
where $\widetilde{\mathcal{G}}_N$ is the following operator on the truncated Fock space $\1^{\le N} \cF(L^2(\R^3))$:
\begin{align} \label{eq:wGN}
 	\widetilde{\mathcal{G}}_N &=	\left(  (N-\cN_+) \int \left(|\nabla\varphi ^2 + \Vext |\varphi |^2 \right) + \frac{1}{2} (N-\cN_+)(N-\cN_+-1)\int \left(V_N \ast \varphi ^2\right)  \varphi ^2 \right) \nn \\
		&\qquad + \left( \sqrt{N-\cN_+} a((-\Delta + V_{\rm ext})\varphi) + {\rm h.c.} \right) \nn\\
		&\qquad+ \left( (N-\cN_+-1) \sqrt{N-\cN_+} a\left(\left(V_N*\varphi^2\right)\varphi\right) + {\rm h.c.} \right)\nn \\ 
		&\qquad+ \left(  \dGamma (-\Delta + V_{\rm ext}) + (N-\cN_+)  \dGamma(V_N*\varphi^2 + N^{-1} K) \right)\nn\\ 
		&\qquad+ \left( \frac{1}{2} \iint  K(x,y)a^*_x a^*_y \,\dd x\, \dd y\,  \frac{\sqrt{(N-\cN_+)(N-\cN_+-1)}}{N} + {\rm h.c.} \right)\nn\\
		&\qquad+ \left( \sqrt{N-\cN_+} \iint  V_N(x-y) \varphi (x) \, a^*_y a_{x'}a_{y'}\,\dd x\, \dd y+ {\rm h.c.} \right)\nn\\
		&\qquad+ \frac{1}{2} \iint  V_N(x-y) a^*_xa^*_y a_{x} a_{y}\,\dd x\, \dd y\nn\\
	&=: {\rm (I)} + {\rm (II)} + {\rm (III)} + {\rm (IV)} + {\rm (V)} + {\rm (VI)} + {\rm (VII)}
\end{align}
where $K$ is the operator on $L^2(\R^3)$ with kernel $K(x,y)=\varphi(x)  NV_N(x-y) \varphi(y)$. Here unlike the presentation in \cite{LNSS-15,LNS-15}, we do not  put the projection $Q$ in the expression of $\widetilde{G}_N$ because we have introduced the projection $\1_{\cF_+}$ in \eqref{eq:UHU-wG}.  
 
 Now let us simplify further $\widetilde{\mathcal{G}}_N$, which is a proper operator on $\1^{\le N}\cF$.

 \subsubsection*{Analysis of {\rm (I)}} Using $ |N-\cN_+|\le N$, we have 
\begin{equation} \label{eq:wGN-1}
	{\rm (I)} \le N \int \left(|\nabla\varphi ^2 + \Vext |\varphi |^2\right) + \frac{N^2}{2} \int \left(V_N \ast \varphi ^2\right)  \varphi ^2.
\end{equation}
 
\subsubsection*{Analysis of {\rm (II)}}  Let us replace $\sqrt{N-\cN_+}$ by $\sqrt{N}$. By the Cauchy--Schwarz inequality  we have
\begin{multline} \label{eq:wGN-2}
	\pm  \left( \left(\sqrt{N-\cN_+}-\sqrt{N}\right) a((-\Delta + V_{\rm ext})\varphi) + {\rm h.c.} \right)\\
	\begin{aligned}[b]
		&\le  N \left(\sqrt{N-\cN_+}-\sqrt{N}\right)^2 + N^{-1} a^*((-\Delta + V_{\rm ext})\varphi) a((-\Delta + V_{\rm ext})\varphi)\\
		&\le N \left( \frac{\cN_+}{\sqrt{N-\cN_+}+\sqrt{N} } \right)^2 + N^{-1}  \cN \norm{\left(-\Delta + V_{\rm ext}\right)\varphi}_{L^2}^2 \le C(\cN^2+1). 
	\end{aligned}
 \end{multline}
 Here we have used $a^*(g)a(g)\le \cN \|g\|_{L^2}^2$ and the fact that $(-\Delta + V_{\rm ext})\varphi \in L^2$.
 
 \subsubsection*{Analysis of {\rm (III)}} We can replace $(N-\cN_+-1)\sqrt{N-\cN_+}$ by $N\sqrt{N}$ as
\begin{align} \label{eq:wGN-3}
	&\pm   \left( \left((N-\cN_+-1) \sqrt{N-\cN_+} -N\sqrt{N}\right) a\left(\left(V_N*\varphi^2\right)\varphi\right) + {\rm h.c.} \right)\nn\\
	&\le N^{-1}   \left((N-\cN_+-1) \sqrt{N-\cN_+} -N\sqrt{N} \right)^2 + N  a^*\left(\left(V_N*\varphi^2\right)\varphi\right)a\left(\left(V_N*\varphi^2\right)\varphi\right)\nn\\
	&\le C\left( \cN_+^2 +1\right) + N  \cN \norm{\left(V_N*\varphi^2\right)\varphi}_{L^2}^2  \le C\left(\cN^2+1\right).
 \end{align}
 Here we have used
\[
 \norm{\left(V_N*\varphi^2\right)\varphi}_{L^2} \le \norm{V_N*\varphi^2}_{L^\infty} \norm{\varphi}_{L^2} \le \norm{V_N}_{L^1}\norm{\varphi^2}_{L^\infty} \norm{\varphi}_{L^2} \le CN^{-1}. 
 \]
 
 \subsubsection*{Analysis of {\rm (IV)}} Similarly to \eqref{eq:K-op} we have $\|K\|_{\rm op} \le C$. Combining with the uniform bound $\|V_N*\varphi^2\|_{L^\infty} \le CN^{-1}$ used above,  we have
\begin{align} \label{eq:wGN-4}
	\pm  (N-\cN_+)  \dGamma(V_N*\varphi^2 + N^{-1} K) \le C\cN. 
\end{align}

  \subsubsection*{Analysis of {\rm (V)}} We can replace $N^{-1}\sqrt{(N-\cN_+)(N-\cN_+-1)}$ by $1$ as
\begin{align} \label{eq:wGN-5}
	&\pm \left(  \iint  K(x,y)a^*_x a^*_y  \left( \frac{\sqrt{(N-\cN_+)(N-\cN_+-1)}}{N}  -1 \right) \,\dd x\, \dd y + {\rm h.c.}\right)  \nn \\
	&\le \iint \left(  |K(x,y)| N^2 \left( \frac{\sqrt{(N-\cN_+)(N-\cN_+-1)}}{N}  -1 \right)^2 + \frac{|K(x,y)|}{N^2}a_x^* a_y^* a_x a_y\right) \dd x \dd y \nn\\
	&\le C \iint \left( NV_N(x-y) \left(\varphi^2(x)+\varphi^2(y) \right) \cN_+^2  + N^{-1} \norm{\varphi}_{L^\infty}^2 V_N(x-y) a_x^* a_y^* a_x a_y\right) \dd x \dd y \nn\\
	&\le  C\cN_+^2 + CN^{-1} \iint V_N (x-y) a_x^* a_y^* a_x a_y \dd x \dd y. 
\end{align}

 \subsubsection*{Analysis of {\rm (VI)}} We can replace $\sqrt{N-\cN_+}$ by $\sqrt{N}$ as
\begin{multline} \label{eq:wGN-6}
	\pm  \left( \left(\sqrt{N}- \sqrt{N-\cN_+}\right) \iint V_N(x-y) (\varphi(x) a_y^* a_x a_y + {\rm h.c.}) \dd x\, \dd y \right) \\
	\begin{aligned}[b]
		&\le \begin{multlined}[t]
			N \norm{\varphi}_{L^\infty}^2 \iint |V_N(x-y)|  a_y^* \left(\sqrt{N}- \sqrt{N-\cN_+}\right)^2 a_y   \dd x\, \dd y \\
			+ N^{-1}\iint V_N(x-y) a_x^* a_y^* a_x a_y \dd x \dd y
		\end{multlined}\\
		&\le C\cN^2 + N^{-1}\iint V_N(x-y) a_x^* a_y^* a_x a_y \dd x \dd y. 
	\end{aligned}
\end{multline}
Here we have used $(\sqrt{N}- \sqrt{N-\cN_+})^2\le  N^{-1}\cN_+ \le \cN_+$ and 
\[
	\int a^*_y \cN_+ a_y \dd y = \int a^*_y a_y (\cN_+ +1) \dd y = \cN (\cN_++1)\le 2 \cN^2. 
\]

\subsubsection*{Conclusion} Inserting \eqref{eq:wGN-1}--\eqref{eq:wGN-6} in \eqref{eq:wGN}, we deduce from \eqref{eq:UHU-wG} that 
\begin{equation} \label{eq:GN-1<}
	\1^{\le N} U_N H_N U_N^* \1^{\le N} \le  \1_{\cF_+}\1^{\le N} ( \mathcal{G}_N + C(\cN+1)^2 )  \1^{\le N} \1_{\cF_+} \quad \text{ on } \cF_+. 
\end{equation}

Now we remove the cut-off $\1^{\le N}$ on the right side of \eqref{eq:GN-1<}. For all terms which are positive and commute with $\cN$, the cut-off $\1^{\le N}$ can be removed for an upper bound. It remains to consider the operator 
\begin{multline*}
	F :=   \sqrt{N} \left(  a \left(  \left(-\Delta + V_{\rm ext} + NV_N*\varphi^2 \right)\varphi \right) + {\rm h.c.}\right) + \frac{1}{2} \iint K(x,y)  (a^*_x a^*_y + {\rm h.c.})   \dd x\, \dd y \\
	+  \sqrt{N} \iint V_N(x-y) \varphi(x) (a^*_y a_{x} a_{y} + {\rm h.c.} ) \dd x \dd y 
\end{multline*} 
on $\cF$. By the Cauchy--Schwarz inequality we can bound
\begin{align} \label{eq:F-F1-CN}
	\pm F &\le \begin{multlined}[t]
		N + \cN \norm{\left(-\Delta + V_{\rm ext} + NV_N*\varphi^2\right)\varphi}_{L^2}^2 +  \iint \left(|K(x,y)|^2 + a_x^*a_y^* a_x a_y \right) \dd x \dd y\\
		+ \iint \left(N |V_N(x-y)|^2 |\varphi(x)|^2 a_y^* a_y + a_x^*a_y^* a_x a_y \right) \dd x \dd y
	\end{multlined}\nn\\
	&\le C(N^3 + \cN^2). 
\end{align}
Denote 
\[ F_1:= F + C_0(N^3 + \cN^2) \ge 0 \quad \textrm{ and } \quad  \1^{>N}=\1-\1^{\le N}. \]
By the Cauchy-Schwarz inequality and \eqref{eq:F-F1-CN} we can bound 
\begin{align*}
	\1^{\le N} F \1^{\le N} - F &= - \1^{\le N} F_1 \1^{>N} - \1^{>N} F_1\1^{\le N} - \1^{>N} F\1^{>N}\\
	&\le N^{-3} \left(\1^{\le N} F_1 \1^{\le N}\right)  + N^3 \left( \1^{>N} F_1 \1^{>N}\right) -  \1^{>N} F\1^{>N}\\
	&\le CN^{-3} \left(N^3+\cN^2\right) \1^{\le N}  + C N^{3} \left(N^3+\cN^2\right) \1^{>N} \le C (\cN +1) ^6. 
\end{align*}
%
Thus in conclusion, we have
\[
	\1^{\le N} \mathcal{G}_N\1^{\le N} \le \mathcal{G}_N + C(\cN+1)^6. 
\]
Inserting this in \eqref{eq:GN-1<} we conclude the proof of Lemma \ref{lem:UHU*}. 
\end{proof}

\subsection{Conclusion of upper bound}

\begin{proof}[Proof of  Lemma \ref{lem:upp}]
Now consider the mixed state $\Gamma_N=U_N^* \1^{\le N} \Gamma \1^{\le N} U_N$. Again we will write $\varphi=\varphi_{\rm GP}$ for short. 

\subsubsection*{Trace normalization} Since 
\begin{equation} \label{eq:1>N-tre}
	\1^{>N}:=\1-\1^{\le N} \le \cN^s N^{-s}, \quad \forall s\ge 1
\end{equation}
we have 
\begin{align} \label{eq:trace-Gamma-N}
	0 &\le 1- \Tr \Gamma_N = 1- \Tr\left(\1^{\le N}\Gamma\right) = \Tr \left(\1^{>N}\Gamma\right) \nn\\
	&\le N^{-s} \tr\left( \cN^s \Gamma  \right) \le  N^{-s} C_s (1+ \Tr\left(\cN \Gamma\right))^s \le C_s N^{-s}, \quad \forall s\ge 1. 
\end{align}
Here we have used \eqref{eq:fluc-N} and $\Tr(\cN \Gamma)= \Tr \gamma\le \Tr k^2 \le C$. 

\subsubsection*{Energy expectation}
Thanks to Lemma \ref{lem:UHU*}, we have
\begin{equation} \label{eq:Tr-HN-GN}
	\Tr\left(H_N \Gamma_N\right)=  \Tr\left( \1^{\le N} U_N H_N U_N^* \1^{\le N} \Gamma\right) \le \Tr\left( \left(\mathcal{G}_N + C(\cN+1)^6\right) \Gamma\right). 
\end{equation}
Using \eqref{eq:fluc-N}  again  we have $\Tr((\cN+1)^6\Gamma)\le C$. Moreover, since $\Gamma$ is a quasi-free state on $\cF_+$ with the one-body density matrices $(\gamma,\alpha)$, by Wick Theorem we have
\begin{align} \label{eq:GN-Gamma}
	\Tr (\mathcal{G}_N \Gamma) &=  N \int \left(|\nabla\varphi |^2 + \Vext |\varphi |^2 \right) + \frac{N^2}{2} \int \left(V_N \ast \varphi ^2\right)  \varphi ^2 \nn\\
	& \quad + \Tr((-\Delta+V_{\rm ext})\gamma)+ \Re N \iint V_N(x-y) \varphi(x) \varphi(y) \alpha(x,y) \dd x \dd y \\
	&\quad + \frac{1+CN^{-1}}{2}\iint V_N(x-y) \left(\gamma(x,x)\gamma(y,y) +|\gamma(x,y)|^2 + |\alpha(x,y)|^2 \right) \dd x \dd y. \nn
\end{align}

It remains to evaluate the right side of \eqref{eq:GN-Gamma} term by term. 

\subsubsection*{Kinetic energy} Using $\gamma=Qk^2Q=(1-P)k^2(1-P)$ we can decompose
\[
	\Tr\left(-\Delta \gamma\right)= \Tr\left(-\Delta k^2\right) + 2\Re  \Tr\left(\Delta P k^2 \right)  + \Tr\left(-\Delta P k^2 P\right).
\]
We have
\begin{align*}
	\Tr\left(-\Delta P k^2 P\right)&= \langle \varphi, -\Delta \varphi\rangle \left\langle \varphi, k^2 \varphi\right\rangle \le C, \\
	\left|\Tr\left(\Delta P k^2\right)\right| &= \left|\langle \varphi, k^2 (\Delta \varphi)\rangle \right| \le \norm{\varphi}_{L^2} \norm{\Delta\varphi}_{L^2} \norm{k^2}_{\rm op} \le C. 
\end{align*}
Now consider the main term $\Tr(-\Delta k^2)$. Similarly to \eqref{eq:K-phi-g}, we write the operator $k$ as 
\[ k=  \varphi(x) N\widehat{\omega_N}(p) \varphi(x)\quad \text{ on } L^2(\R^3) \]
where  $\omega_N=1-f_N$ and $\varphi(x)$, $\widehat{\omega_N}(p)$ are multiplication operators on the position and momentum spaces. By the IMS formula \eqref{eq:IMS-formula} we can decompose  
\begin{align*}
	\Tr(-\Delta k^2)&= N^2 \tr\left(  \varphi(x)  p^2 \varphi(x)  \widehat{\omega_N}(p)  \varphi^2(x) \widehat{\omega_N}(p)  \right)\\
	&=\begin{multlined}[t]
		\frac{N^2}{2} \Tr \left( \left(\varphi^2(x) p^2 + p^2 \varphi^2(x)\right)  \widehat{\omega_N}(p)  \varphi^2(x) \widehat{\omega_N}(p)  \right) \\
		+ N^2 \tr\left(  |\nabla \varphi(x)|^2  \widehat{\omega_N}(p)  \varphi^2(x) \widehat{\omega_N}(p)  \right).
	\end{multlined}
\end{align*}
The first term can be computed exactly using the scattering equation \eqref{eq:scatering-eq-N-omega} and \eqref{eq:tr-int}: 
\begin{multline*}
	\frac{N^2}{2} \Tr \left( \left(\varphi^2(x) p^2 + p^2 \varphi^2(x)\right)  \widehat{\omega_N}(p)  \varphi^2(x) \widehat{\omega_N}(p)  \right)\\
	=\frac{N^2}{2} \Re \Tr \left( \varphi^2(x)    \widehat{V_Nf_N}(p)  \varphi^2(x) \widehat{\omega_N}(p)  \right) = \frac{N^2}{2} \int N \left(\left(V_Nf_N \omega_N\right)*\varphi^2\right) \varphi^2. 
\end{multline*} 
The second term can be bounded using  \eqref{eq:tr-int}
\[ N^2 \tr\left( |\nabla \varphi(x)|^2  \widehat{\omega_N}(p)  \varphi^2(x) \widehat{\omega_N}(p)  \right)= N^2  \int \left(\omega_N^2*\varphi^2\right) |\nabla \varphi|^2 \le C \norm{\nabla \varphi}_{L^2}^4. \]
Here we have used 
\begin{equation} \label{eq:omega*phi2}
	N^2  \norm{\omega_N^2*\varphi^2}_{L^\infty} \le C  \norm{|x|^{-2}* \varphi^2}_{L^\infty} \le  C \norm{\nabla \varphi}_{L^2}^2
\end{equation}
by \eqref{eq:bounds_on_w} and Hardy's inequality $|x|^{-2}\le 4(-\Delta)$. Thus 
\begin{equation} \label{eq:GN-Gamma-1}
	\Tr(-\Delta \gamma)= \frac{N^2}{2} \int N \left(\left(V_Nf_N (1-f_N)\right) *\varphi^2\right) \varphi^2 + O(1). 
\end{equation}

\subsubsection*{External potential energy} Using \eqref{eq:omega*phi2} again, we have
\begin{align*}
	|(k^2)(x,y)|&= \left| \int \dd z k(x,z) k(z,y) \right| = N^2 \left| \int \dd z  \varphi(x)  \omega_N(x-z) \varphi^2(z) \omega_N(z-y)  \varphi(y) \right| \\
	&\le  \varphi(x) \varphi(y)  \frac{N^2}{2}  \int \dd z  \left(\omega_N^2(x-z)+ \omega_N^2(z-y)\right) \varphi^2(z)  \le C \varphi(x) \varphi(y). 
\end{align*}
Hence, 
\begin{align} \label{eq:gamma-x-y}
	|\gamma(x,y)| &= \left|\left((\1-P) k^2 (\1-P)\right)(x,y)\right| \nn\\
	&=\begin{multlined}[t]
		\left| \left(k^2\right)(x,y) - \varphi(x) \int \dd z \varphi(z) \left(k^2\right)(z,y)   - \varphi(x)  \int \dd z \varphi(z) \left(k^2\right)(x,z)\right. \\
		\left.+ \varphi(x)\varphi(y)  \iint \dd x \dd y \varphi(x)\varphi(y) \left(k^2\right)(x,y) \right|
	\end{multlined}\nn\\
	&\le C  \varphi(x) \varphi(y).
\end{align}
Consequently, 
\begin{align} \label{eq:GN-Gamma-2}
	\Tr(V_{\rm ext} \gamma) = \int V_{\rm ext}(x) \gamma(x,x) \dd x \le C \int V_{\rm ext} \varphi^2 \le C. 
\end{align}

\subsubsection*{Bogoliubov pairing energy} Since $\alpha(x,y)=(Q\otimes Q k)(x,y)$, by decomposing $Q=\1-P$ we have
\begin{align} \label{eq:alpha-x-y}
	|\alpha(x,y) - k(x,y)| &=
	\begin{multlined}[t]
		\left| - \varphi(x) \left(\varphi^2*N\omega_N\right)(y) \varphi(y)  - \varphi(x)    \left(\varphi^2*N\omega_N\right)(x) \varphi(y)\vphantom{\int}\right.\\
		\left.+ \varphi(x)\varphi(y)  \int \left(N\omega_N* \varphi^2\right)\varphi^2 \right|
	\end{multlined}\nn\\
	&\le C \varphi(x) \varphi(y). 
\end{align}
Here we have used 
\[ N\norm{\omega_N*\varphi^2}_{L^\infty} \le C \norm{|x|^{-1}*\varphi^2}_{L^\infty}\le C \]
which is similar to \eqref{eq:omega*phi2}. Thus
\begin{align} \label{eq:GN-Gamma-3}
	&\Re N \iint V_N(x-y) \varphi(x) \varphi(y) \alpha(x,y) \dd x \dd y \nn\\
	&\le \Re N \iint V_N(x-y) \varphi(x) \varphi(y) \big(k(x,y) + C\varphi(x)\varphi(y)\big) \dd x \dd y \nn\\
	&= - N^2 \iint (V_N\omega_N)(x-y) \varphi^2(x) \varphi^2(y) \dd x \dd y + C N \iint V_N(x-y) \varphi^2(x) \varphi^2(y) \dd x \dd y \nn\\
	&= - N^2 \int \left((V_N(1-f_N))*\varphi^2\right)\varphi^2 +O(1). 
\end{align}
In the last estimate we have used 
\begin{equation} \label{eq:GN-Gamma-easyyy}
	N \iint \dd x \dd y V_N(x-y) \varphi^2(x) \varphi(y)^2 \le C \norm{\varphi}_{L^\infty}^2  \norm{\varphi}_{L^2}^2  N\norm{V_N}_{L^1}\le C. 
\end{equation}

\subsubsection*{Interaction energy} Using \eqref{eq:gamma-x-y} and \eqref{eq:GN-Gamma-easyyy} we have
\begin{equation} \label{eq:GN-Gamma-4}
	\iint \dd x \dd y V_N(x-y) \left( \gamma(x,x) \gamma(y,y) + |\gamma(x,y)|^2\right) \le CN^{-1}.
\end{equation}
Moreover,  by \eqref{eq:alpha-x-y} and the Cauchy--Schwarz inequality, 
\[
	|\alpha(x,y)|^2  \le (1+N^{-1}) |k(x,y)|^2 + CN |\varphi(x)|^2 |\varphi(y)|^2. 
\]
Therefore,
\begin{equation} \label{eq:GN-Gamma-5}
	\iint V_N(x,y) |\alpha(x,y)|^2 \dd x \dd y \le  N^2 \int \left(\left(V_N (1-f_N)^2\right) *\varphi^2\right)  \varphi^2 + C.  
\end{equation}

\subsubsection*{Conclusion} Inserting \eqref{eq:GN-Gamma-1}, \eqref{eq:GN-Gamma-2}, \eqref{eq:GN-Gamma-3}, \eqref{eq:GN-Gamma-4}, \eqref{eq:GN-Gamma-5} in \eqref{eq:Tr-HN-GN}-\eqref{eq:GN-Gamma} we find that 
\begin{align*}
	\Tr(H_N\Gamma_N)& \le N \int \left(|\nabla\varphi|^2 + \Vext \varphi^2\right) + \frac{N^2}{2} \int \left(V_N \ast \varphi^2\right)  \varphi^2 \nn\\
	& \quad + \frac{N^2}{2}  \int \left((V_N f_N(1-f_N))*\varphi^2\right)\varphi^2  -  N^2 \int \left((V_N (1-f_N))* \varphi^2\right) \varphi^2 \\
	&\quad +\frac{N^2}{2}\int \left((V_N (1-f_N)^2) *\varphi^2\right)  \varphi^2 + C\\
	&= N \int \left(|\nabla\varphi|^2 + \Vext \varphi^2 \right) + \frac{N^2}{2} \int \left((V_N f_N)  \ast \varphi^2\right)  \varphi^2 +C.
\end{align*}
Finally, by Young's inequality we have
\begin{align*}
	\frac{N^2}{2} \int \left(V_N f_N  \ast \varphi^2\right)  \varphi^2 &\le \frac{N^2}{2} \norm{(V_Nf_N)*\varphi^2}_{L^2(\R^3)} \norm{\varphi^2}_{L^2(\R^3)}\\
	&\le \frac{N^2}{2} \norm{V_Nf_N}_{L^1(\R^3)} \norm{\varphi^2}_{L^2(\R^3)}^2= 4\pi \ao N \int \varphi^4. 
\end{align*}
Thus
\begin{align*}
\Tr(H_N\Gamma_N) \le N \int \left(|\nabla\varphi|^2 + \Vext \varphi^2 \right) + 4\pi \ao N \int  \varphi^4 + C = N e_{\rm GP} +C.
\end{align*}
Finally, by the variational principle
\[
	E_N = \inf {\rm Spec}(H_N) \le \frac{\Tr(H_N\Gamma_N) }{\Tr \Gamma_N} \le \frac{N e_{\rm GP} +C}{ 1+CN^{-1}} \le N e_{\rm GP} + O(1).
\]
This ends the proof of Lemma \ref{lem:upp}. 
\end{proof}
The proof of Theorem \ref{thm:main} is complete.

\end{document}